\documentclass{llncs}

\usepackage{url}
\usepackage{amsmath,amssymb}
\usepackage{mathtools}
\usepackage{mathpartir}
\usepackage{stmaryrd}
\usepackage{listings}
\usepackage{xcolor}
\usepackage{graphicx}
\usepackage{booktabs}
\usepackage{cite}
\usepackage{multirow}
\usepackage{mathrsfs}
\usepackage{tikz}
\usepackage{enumitem}
\usepackage[colorlinks=true, linkcolor=blue!70!black, citecolor=blue!70!black, urlcolor=blue!70!black]{hyperref}
\usepackage{verified-badges}
\setlist{itemsep=0.5em}
\setlist[enumerate,1]{label=\textup{(\arabic*)}}
\usetikzlibrary{arrows.meta,positioning,shapes}

\lstdefinelanguage{MSC}{
  keywords={lifeline,workflow,action,act,msg,skip,send,recv,loop,if,else,var,kind,system,user,parse,decode,max,while,until,true,false,return},
  keywordstyle=\bfseries\color{blue!70!black},
  identifierstyle=\color{black},
  commentstyle=\itshape\color{gray},
  stringstyle=\color{red!70!black},
  morestring=[b]",
  morecomment=[l]{//},
  sensitive=true,
}

\AtBeginDocument{}
\newcommand{\llangle}{{\langle\hspace{-0.2em}\langle}}
\newcommand{\rrangle}{{\rangle\hspace{-0.2em}\rangle}}

\newcommand{\sem}[1]{\llbracket #1 \rrbracket}
\newcommand{\sempref}[1]{\llbracket #1 \rrbracket_{\mathsf{pref}}}
\newcommand{\ltr}[2]{\llangle #2 \rrangle_{#1}}
\newcommand{\lpref}[2]{\llangle #2 \rrangle_{#1,\mathsf{pref}}}

\newcommand{\proj}[2]{\pi_{#1}(#2)}

\newcommand{\df}{=}

\newcommand{\Work}{\mathsf{Work}}

\newcommand{\LS}[1]{\mathscr{L}(#1)}

\newcommand{\IfT}[2]{\mathsf{if}_{\top}(#1@#2)}
\newcommand{\IfF}[2]{\mathsf{if}_{\bot}(#1@#2)}
\newcommand{\WhT}[2]{\mathsf{while}_{\top}(#1@#2)}
\newcommand{\WhF}[2]{\mathsf{while}_{\bot}(#1@#2)}

\newcommand{\ActF}[3]{\mathtt{act}\ #1 : #2 = #3}

\newcommand{\enriched}[1]{\hat{#1}}

\begin{document}
\pagestyle{plain}

\title{Provable Coordination for LLM Agents via Message Sequence Charts}

\author{Benedikt~Bollig\inst{1}\orcidID{0000-0003-0985-6115} \and
Matthias~F{\"u}gger\inst{1}\orcidID{0000-0001-5765-0301}
\and
Thomas~Nowak\inst{1,2}\orcidID{0000-0003-1690-9342}
}
\authorrunning{B.~Bollig, M.~F{\"u}gger, T.~Nowak}
%
\institute{Université Paris-Saclay, CNRS, ENS Paris-Saclay, LMF, Gif-sur-Yvette, France
\and
Institut Universitaire de France, Paris, France}

\maketitle

\begin{abstract}
Multi-agent systems built on large language models (LLMs) are difficult to reason about. Coordination errors such as deadlocks or type-mismatched messages are often hard to detect through testing. We introduce a domain-specific language for specifying agent coordination based on message sequence charts (MSCs). The language separates message-passing structure from LLM calls, tool calls, and human control points, whose outcomes remain unpredictable. We define the syntax and semantics of the language and present a syntax-directed projection that generates deadlock-free local agent programs from global coordination specifications. We illustrate the approach with a diagnosis consensus protocol and show how coordination properties can be established independently of LLM nondeterminism. We also describe a runtime planning extension in which an LLM dynamically generates a coordination workflow for which the same structural guarantees apply. An open-source Python implementation of our framework is available as ZipperGen.

\end{abstract}

\section{Introduction}
\label{sec:intro}

Multi-agent systems built on large language models (LLMs) are increasingly common. They combine the capabilities of individual AI components with the structural advantages of distributed systems, such as concurrency, separation of concerns, and explicit communication patterns. A typical example is a group of LLM agents that iteratively exchange assessments until they reach agreement on a task, such as giving a medical diagnosis.

These systems are, however, hard to design and to reason about.
While classical distributed systems are already notoriously error prone, the stochastic nature of LLM outputs further complicates reasoning about such systems. When agents exchange messages and collaborate, problems such as deadlocks or lost messages are difficult to rule out through testing alone.
Frameworks like LangGraph~\cite{langgraph}, AutoGen~\cite{wu2023autogen}, and CrewAI~\cite{crewai} provide convenient infrastructure for building such systems, but, to our knowledge, do not provide such formal coordination guarantees. The broader engineering challenges of such systems, such as transparency, traceability, and error attribution, are well-documented~\cite{dibia2025multiagent}. Applying formal methods to this setting is therefore natural.

\paragraph{Contributions.} For this purpose, we propose a domain-specific language (DSL) for agent coordination based on message sequence charts (MSCs)~\cite{itu-msc}, an International Telecommunication Union (ITU) standard for communication protocols. MSCs combine a visual notation with a formal semantics. They represent both the local progress of individual agents (their lifelines) and the messages exchanged between them, making causal dependencies in a system execution easy to track.

Abstractly, an MSC describes a single execution. Higher-level formalisms compose simple MSCs into sets (or languages). They can be thought of as global workflows: specifications of what a distributed system should do, independently of how it is implemented. Much research has addressed the question of how such global specifications can be realized in a distributed fashion. For MSC-based formalisms, this belongs to a broader automata-theoretic study of realizability and verification~\cite{AlurEY05,GenestMSZ06,BolligFG25}, where global specifications are compared with distributed implementations in the form of local machines. A conceptually clean way to obtain such implementations is projection: deriving a local program for each participant by restricting the global specification to that participant's view. MPST offers a related projection discipline from global to local types~\cite{HondaYC16}, with recent work also developing automata-based projection techniques~\cite{LiSWZ23}.

Our approach works with structured programs rather than automata and targets LLM components directly. Workflow-level specifications suit LLM agents well: rather than following a finite transition graph, agents maintain variable-length context, evaluate local conditions, and branch based on their own state. They also make oversight points explicit: a workflow can state that a tool call or external action is causally after a human approval, and projection preserves this order in the distributed implementation. Moreover, the projected local programs are programs again: they can be embedded in larger workflows, exposed as APIs, or composed with other components without further translation. To our knowledge, the above frameworks do not address the synthesis of provably correct distributed agent programs.

The same holds when the coordination structure is not fixed in advance but generated at runtime by a planning LLM. Asking such a system to produce deadlock-free distributed code directly is unreasonable: it would have to ensure that every send by one agent is matched by a receive of another, in every branch and under every execution order. With our language, the planner needs only to produce a syntactically valid, well-typed global workflow.
Deadlock-freedom is then guaranteed by the projection. The language is therefore not only a programming interface for humans, but also a natural target for LLM-based planning.

We extend standard MSCs with two constructs: action blocks, which treat LLM calls, tool calls, and human inputs as opaque typed functions, and owned control flow, where each conditional and loop has an explicit decider.
The first extension separates what can be guaranteed statically, namely
coordination correctness and deadlock-freedom of the message-passing structure,
from what cannot: the correctness of LLM outputs, tool results, or overall task
success. The second makes
projection syntax-directed: it is a structural recursion on the global
workflow, and the correctness proof is a structural induction. No automaton
is constructed at any point.

\paragraph{Related work.}
Our work connects two lines of research.
The first is the formal side: realizability and implementability for MSCs, communicating
automata, and related global protocol models~\cite{AlurEY05,Genest05,GenestMSZ06,BolligFG25,Lohrey03,LiW25,LiSWZ25,Stutz23,GiustoLU25} address the question of whether a global specification admits a correct distributed implementation.
Rather than solving this problem post-hoc, we follow a by-construction approach:
the global language is designed so that projection is always syntax-directed and yields
well-structured local programs directly.
Closest to our approach are~\cite{Genest05,AbdallahGHJ13}. The work~\cite{Genest05}
automatically adds the messages needed to implement global specifications with local
asynchronous controllers. In~\cite{AbdallahGHJ13} non-local
HMSCs are made local by adding messages and processes so that standard implementation techniques
can be applied. Both work with automaton-based global specifications and are not
designed for LLM-based agents.

Multiparty session types (MPSTs)~\cite{HondaVK98,HondaYC16,LiSWZ23,DenielouY12,abs-1203-0780,FossatiHY14} also project a global description to local views, but require well-formedness conditions for projection to be defined, notably that uninvolved participants have consistent local behaviors across branches.
Scribble~\cite{YoshidaHNN13} is a prominent practical MPST-based protocol language, and its refinement-typed extension~\cite{HNY20} further adds data-level verification, an orthogonal direction we do not pursue.
In contrast, our language is designed so that projection is always defined for its constructs: owned control flow together with explicit control broadcasts eliminates the need for mergeability checks on uninvolved participants in standard MPST-style projection~\cite{DenielouY12,HNY20}.
Process calculi~\cite{csp,MilnerPW92a} work in the opposite direction: local processes are composed to obtain global behavior, rather than projected from a global specification.

Choreographic programming~\cite{CarboneM13,Montesi2023}
also starts from a global program and derives one local program per participant.
Its syntax is often based on Alice-and-Bob-style descriptions of communication.
ZipperGen follows this syntactic and global-to-local style, using MSC-style
workflows as its source language.
This line includes asynchronous global programming with deadlock-freedom by
construction, as well as functional approaches such as
Pirouette~\cite{HirschG22} and HasChor~\cite{ShenKK23}. HasChor is particularly
close: during projection, its conditional construct broadcasts the condition
value, so that participants know which branch to follow.
Pact~\cite{GopinathanFNTB26} is complementary recent work on choreographies for
agentic systems. Our focus is instead on structural projection guarantees for
LLM-agent coordination.

The second line concerns practical multi-agent frameworks such as LangGraph, AutoGen,
and CrewAI~\cite{langgraph,wu2023autogen,crewai}.
These treat coordination primarily as an engineering concern, with correctness assessed through testing, tracing, or runtime monitoring rather than static projection.
We separate the coordination layer from the LLM layer and reason about the former independently.
In particular, deadlock-freedom holds even for runtime-generated workflows, which to our knowledge the above frameworks do not guarantee.

\paragraph{Outline.}
Section~\ref{sec:motivation} presents introductory examples.
Sections~\ref{sec:syntax}--\ref{sec:distributed} define the syntax and semantics of global and local programs.
Section~\ref{sec:implementation} presents the projection, proves correctness and deadlock-freedom, and describes the ZipperGen implementation (\url{https://zippergen.io}).
Section~\ref{sec:examples} illustrates the framework.
Section~\ref{sec:conclusion} concludes.
Proofs are deferred to the appendix.
All definitions have been formalized and all results have been machine-checked in the Lean 4 proof assistant \cite{Moura2021}. Clicking a badge \leanformalized{} or \leanproof{} opens the corresponding Lean source file.

\paragraph*{Acknowledgments.}

This work was partly supported by the ``France 2030'' government investment plan
managed by ANR, under the reference ANR-23-PEIA-0006.

\section{Introductory Examples}
\label{sec:motivation}

We illustrate our language on a simple multi-agent workflow.

\begin{lstlisting}[caption={Review workflow (lifeline and most variable declarations omitted).},label={lst:review-before-action}]
workflow reviewed_execution(task: str @ Planner) -> str {
    var critique: str = "no review" @ Orchestrator
    act Planner : (plan, plan_needs_review) = make_plan(task)

    if plan_needs_review@Planner then {
        msg Planner(plan) -> Reviewer(plan)
        act Reviewer : critique = review_plan(plan)
        msg Reviewer(critique) -> Orchestrator(critique)
    } else {
        act Planner : review_skipped = record_no_review(plan)
    }

    msg Planner(plan) -> Executor(plan)
    act Executor : result = execute_plan(plan)
    msg Executor(result) -> Orchestrator(result)
    act Orchestrator : summary = finalize(critique, result)
    return summary @ Orchestrator
}
\end{lstlisting}

The program in Listing~\ref{lst:review-before-action} specifies a
review-and-execution workflow with four lifelines:
$\mathit{Planner}$, $\mathit{Reviewer}$, $\mathit{Executor}$, and
$\mathit{Orchestrator}$.
Planner receives a task, produces a plan, and decides whether it needs
review.
In the first scenario ($\mathit{plan\_needs\_review}$ is true), Planner
sends the plan to a human Reviewer, who critiques it and forwards the critique
to Orchestrator.
In the second scenario, the review is skipped.
In both cases, Planner forwards the plan to Executor, whose result is
sent to Orchestrator, which finalizes the outcome.

Figure~\ref{fig:msc-review} gives a graphical view of the workflow,
abstracting away variable declarations and types.
The dashed box separates the two branches of the conditional;
the figure thus represents two MSCs, one per scenario.
Action blocks such as
$\ActF{\mathit{Planner}}{(\mathit{plan},\mathit{plan\_needs\_review})}{\mathit{make\_plan}(\mathit{task})}$
model opaque computations. They may be LLM calls, tool calls, or human actions;
in this example, $\mathit{review\_plan}$ is a human review action.
The conditional
$\mathtt{if}\ \mathit{plan\_needs\_review}@\mathit{Planner}$ has an owned
decider: only $\mathit{Planner}$ chooses the branch, and
$\mathit{Reviewer}$ participates only in the true branch.
The last action depends on both $\mathit{critique}$ and
$\mathit{result}$. Thus,
computations remain opaque, but the coordination around them is made
explicit.
When review is required, review and execution may proceed concurrently;
$\mathit{Orchestrator}$ nevertheless receives the critique before finalizing the
result. When review is skipped, $\mathit{Planner}$ records this explicitly before
execution continues.
The \texttt{return} statement designates the output lifeline. It is not part of the MSC but specifies which lifeline's value is returned to the caller.
Note that the critique is sent to $\mathit{Orchestrator}$ rather than back to $\mathit{Planner}$: the review influences the final outcome but does not block execution.

\begin{figure}[t]
\centering
\begin{tikzpicture}[scale=0.9,thick,
  action/.style={circle, fill=black, minimum size=5pt, inner sep=0pt},
  lbl/.style={font=\scriptsize},
  msg/.style={-{Triangle[length=4pt,width=4pt]}, line width=0.8pt}]

\def\xP{0} \def\xR{2.6} \def\xO{5.2} \def\xE{7.8}

\node[draw, minimum width=1.4cm, minimum height=0.55cm, font=\scriptsize] at (\xP, 0) {Planner};
\node[draw, minimum width=1.4cm, minimum height=0.55cm, font=\scriptsize] at (\xR, 0) {Reviewer};
\node[draw, minimum width=1.4cm, minimum height=0.55cm, font=\scriptsize] at (\xO, 0) {Orchestr.};
\node[draw, minimum width=1.4cm, minimum height=0.55cm, font=\scriptsize] at (\xE, 0) {Executor};

\draw (\xP, -0.28) -- (\xP, -7.4);
\draw (\xR, -0.28) -- (\xR, -7.4);
\draw (\xO, -0.28) -- (\xO, -7.4);
\draw (\xE, -0.28) -- (\xE, -7.4);

\node[action, label={[lbl]right:\itshape make\_plan}] at (\xP, -0.6) {};

\draw[dashed, rounded corners=3pt] (-1.0, -1.3) rectangle (8.8, -4.6);
\node[lbl, anchor=south west] at (-1.0, -1.3) {[\textit{plan\_needs\_review}]};

\draw[msg] (\xP, -1.9) -- (\xR, -1.9) node[midway, lbl, above] {plan};
\node[action, label={[lbl]right:\itshape review\_plan}] at (\xR, -2.5) {};
\draw[msg] (\xR, -3.2) -- (\xO, -3.2) node[midway, lbl, above] {critique};

\draw[dashed, line width=0.6pt] (-1.0, -3.7) -- (8.8, -3.7);
\node[lbl, anchor=south west] at (-1.0, -3.7) {[\textit{else}]};

\node[action, label={[lbl]right:\itshape no review}] at (\xP, -4.15) {};

\draw[msg] (\xP, -5.2) -- (\xE, -5.2) node[midway, lbl, above] {plan};
\node[action, label={[lbl]left:\itshape execute\_plan}] at (\xE, -5.9) {};
\draw[msg] (\xE, -6.5) -- (\xO, -6.5) node[midway, lbl, above] {result};
\node[action, label={[lbl]right:\itshape finalize}] at (\xO, -7.1) {};

\end{tikzpicture}
\caption{MSC for \texttt{reviewed\_execution}. Filled dots are local actions (opaque computations), arrows are messages labelled by the value sent. Planner owns the conditional. The true branch involves Reviewer. In the false branch, Planner records that review was skipped. After the conditional, Planner forwards the plan to Executor. In the true branch, Reviewer and Executor may therefore proceed concurrently. Orchestrator finalizes the result.}
\label{fig:msc-review}
\end{figure}

The program above is written from a global viewpoint: it describes the full
coordination of all lifelines in a single text.
In practice, each agent runs its own local program and communicates with
the others exclusively by sending and receiving messages over first-in, first-out (FIFO) channels.
FIFO channels allow agents to
overlap their work. A sender can deposit a message and continue immediately, and
the receiver picks it up whenever its local execution reaches the corresponding
receive. This allows agents that do not depend on each other to proceed in parallel.

One obvious implementation strategy is simple polling: each agent repeatedly checks its
incoming channels and processes the first available message. But this is
not sufficient even for programs satisfying the local property of~\cite{GenestMSZ06}
or the local-choice property of~\cite{Ben-AbdallahL97,HelouetJ00}. Consider the true branch of \texttt{reviewed\_execution}.
Once $\mathit{Planner}$ has sent the plan to $\mathit{Reviewer}$, its
participation in the if-block is complete, and it immediately sends the plan to
$\mathit{Executor}$. As a result, $\mathit{Reviewer}$ and $\mathit{Executor}$
run concurrently and both send their outputs to $\mathit{Orchestrator}$.
Because the channels are asynchronous, $\mathit{Executor}$ may finish and
deliver $\mathit{result}$ to $\mathit{Orchestrator}$ before $\mathit{Reviewer}$
has even dequeued the plan.
This prefix execution is reachable under polling, yet it cannot be extended to
a valid completion: $\mathit{Orchestrator}$ dequeues $\mathit{result}$ before
$\mathit{critique}$ arrives. Our projection overcomes this by
having $\mathit{Planner}$ broadcast an explicit control message that tells
$\mathit{Orchestrator}$ which branch was taken.
$\mathit{Orchestrator}$'s local program then enforces the correct receive order
(cf.~Section~\ref{sec:implementation}).

In fact, such control messages are necessary in general.
Consider the program below, where $A$ repeatedly
tosses a coin and $B$ performs a local step in each iteration:
\[
 \mathtt{while}\ \mathit{heads}\texttt{@}A\ \mathtt{do}\ \{\,
    \mathtt{act}\ A: \mathit{heads} = \mathit{toss}()\ ;\
    \ActF{B}{\mathit{b}}{\mathit{step}()} \\
    \,\}\ \mathtt{exit}\ \{\,\varepsilon\,\}
\]
$A$ and $B$ must execute the same number of iterations, yet the iteration
count is determined solely by $A$'s coin tosses, which are private to $A$.
There is no user message that $B$ could observe to recover it.
Any correct distributed implementation must therefore send $B$ an explicit
control signal at each iteration, exactly as our projection does.%

Most work on MSC realizability, e.g., \cite{AlurEY05,Lohrey03,LiSWZ25,GenestMSZ06,LiW25,Stutz23}, formulates the question as a decision problem (given a global specification, does a deadlock-free distributed implementation exist?) or identifies sufficient syntactic criteria for it. In practice, one wants simply to write a program and have it run. The constraint is even less realistic when workflows are generated at runtime by a planning LLM, which cannot be expected to produce only specifications satisfying a non-trivial syntactic condition. Our language is designed with this in mind.

\section{Global Workflows}
\label{sec:syntax}

A global workflow describes the full coordination structure of a
multi-agent system in a single, centralized specification.
Rather than programming each agent separately, the developer writes one workflow
that captures which agents communicate, in what order, and who owns each
control decision.
This section defines the formal syntax and semantics of such global workflows.

\subsection{Syntax}

Throughout this section, fix a finite set $\mathscr{L}$ of lifelines and a set
$\mathscr{F}$ of action symbols. We use $A,B\in\mathscr{L}$ for lifelines,
$f\in\mathscr{F}$ for local computations such as LLM calls or deterministic
helper functions, $\vec{x},\vec{y}$ for payload tuples, and $c$ for Boolean
conditions over local variables of designated lifelines.
Communication proceeds over directed unbounded FIFO channels: for every ordered pair
$(A,B)$ of distinct lifelines there is one channel $A\to B$.

\begin{definition}[Abstract syntax \leanformalized{https://github.com/zippergen-io/zippergen-lean/blob/isola-camera-ready/isola/MSCAgents/Syntax.lean}]
\label{def:phase-syntax}
The set $\Work$ of
global workflows is generated by the following grammar:
\begin{align*}
P ~::=~ &\varepsilon
~\mid~
\mathtt{msg}\ A(\vec{x}) \to B(\vec{y})
~\mid~
\ActF{A}{\vec{y}}{f(\vec{x})}
~\mid~ P_1 ; P_2
~\mid~\\
&\mathtt{if}\ c@B\ \mathtt{then}\ P_\top\ \mathtt{else}\ P_\bot
~\mid~
\mathtt{while}\ c@B\ \mathtt{do}\ P_{\mathit{body}}\ \mathtt{exit}\ P_{\mathit{exit}}
\end{align*}
For message statements, we require distinct endpoints: in
$\mathtt{msg}\ A(\vec{x}) \to B(\vec{y})$ we assume
$A \neq B$ (no self channels).
Payload components in $\vec{x} = (x_1, \ldots, x_m)$ and $\vec{y}=(y_1, \ldots, y_n)$ may be variables
or concrete values, such as the truth values $\top$ (true) and $\bot$ (false).
In action statements $\ActF{A}{\vec{y}}{f(\vec{x})}$, the outputs $\vec{y}$ must be variables (assignment targets), while the inputs $\vec{x}$ may be variables or constants.
Each lifeline maintains its own local variable store.
Variable names are scoped per lifeline,
so distinct lifelines may use the same name independently.
\end{definition}

We emphasize that the grammar contains no return statement.
A return is not a coordination step: it does not send a message, perform a
computation, or affect any other lifeline.
It names the lifeline that produces the final value and the variable holding it.
We therefore treat it as a declaration outside the formal language and it plays no
role in the projection or the correctness proofs.
In the ZipperGen implementation, \texttt{return}~\texttt{var}~\texttt{@}~$A$
serves as this declaration and is checked syntactically.

The statements have the following intuitive meaning (Listing~\ref{lst:review-before-action}
illustrates all of them except the while construct):
\begin{itemize}
  \item $\mathtt{msg}\ A(\vec{x}) \to B(\vec{y})$: a communication step in which $A$ sends a payload tuple and $B$ receives it, possibly binding components to local variables.
  \item $\ActF{A}{\vec{y}}{f(\vec{x})}$: a local computation performed by lifeline $A$; in Listing~\ref{lst:review-before-action}, examples are $\mathit{make\_plan}$, $\mathit{review\_plan}$, and $\mathit{finalize}$.
  \item $\mathtt{if}\ c@B\ \mathtt{then}\ P_{\top}\ \mathtt{else}\ P_{\bot}$: a conditional whose branch is chosen exclusively by lifeline $B$. The annotation $@B$ records this ownership explicitly, even though the branch programs $P_{\top}$ and $P_{\bot}$ may themselves involve several lifelines.
  \item $\mathtt{while}\ c@B\ \mathtt{do}\ P_{\mathit{body}}\ \mathtt{exit}\ P_{\mathit{exit}}$: an owned loop in which $B$ decides whether another iteration of $P_{\mathit{body}}$ is executed or whether control exits through $P_{\mathit{exit}}$. Again, the body and exit programs may involve multiple lifelines.
\end{itemize}

\subsection{Semantics}
\label{sec:semantics}

The semantics assigns to each global workflow a set of MSCs, each describing one
possible global execution. Figure~\ref{fig:msc-review} gives the intended
intuition, but it should be read as depicting two MSCs, one for each branch of
the conditional shown by the dashed case distinction. In general, an execution
is represented by one vertical word per lifeline, together with FIFO-matching
edges between sends and receives on each channel. Local actions appear as
internal events on a single lifeline, while owned conditionals and loops
contribute explicit choice events recording which branch or loop outcome was
selected. We now make this precise by defining MSCs as tuples of
local words.

\begin{definition}[Local Alphabets \leanformalized{https://github.com/zippergen-io/zippergen-lean/blob/isola-camera-ready/isola/MSCAgents/Alphabets.lean}]
\label{def:alphabets}
For each lifeline $A\in\mathscr{L}$, the (infinite) local-word alphabet
$\Sigma_A$ contains:
\begin{itemize}
  \item send letters $\mathtt{send}\ A(\vec{x})\to B$ and receive letters $\mathtt{recv}\ A(\vec{y})\leftarrow B$,
        where $\vec{x}$ and $\vec{y}$ range over tuples of constants and variables,
  \item local action letters $\ActF{A}{\vec{y}}{f(\vec{x})}$, where $\vec{x}$ ranges over tuples of constants and variables and $\vec{y}$ over tuples of variables,
  \item choice letters $\IfT{c}{A},\IfF{c}{A},\WhT{c}{A},\WhF{c}{A}$ for condition~$c$.
\end{itemize}
\end{definition}

We let $\Sigma_A^*$ denote the set of finite words over $\Sigma_A$,
which includes the empty word $\varepsilon$. We use $u \cdot v$ or $uv$ for the concatenation of words $u$ and $v$.

Though an MSC will be defined as a tuple of local words, one must also specify which send
events are matched with which receive events across each channel. Since our
execution model uses FIFO channels, this matching is determined canonically by
pairing the first send on a channel with the first receive on that channel, the
second send with the second receive, and so on. The next definition extracts
this channel-wise FIFO matching from a tuple of local words.

\begin{definition}[Tuple FIFO Relations \leanformalized{https://github.com/zippergen-io/zippergen-lean/blob/isola-camera-ready/isola/MSCAgents/FIFORelation.lean}]
\label{def:tuple-fifo}
Fix a tuple of local words $M=(w_A)_{A\in\mathscr L}$ with
$w_A\in\Sigma_A^*$ for each $A\in\mathscr L$.
We write $|w_A|$ for the length of $w_A$ and $w_A[j]$ for its $j$th symbol
(starting at index 1). The event set of $M$ is
\[
E_M\df\{(A,i)\mid A\in\mathscr L,\ 1\le i\le |w_A|\}.
\]

For each channel $A\to B$, define the send-event sequence
$
\mathsf{snd}_{A\to B}(M)=(s_1,\ldots,s_m)
$
where $(s_i)_i$ are exactly the events $(A,j)\in E_M$ such that
$w_A[j]$ is of the form $\mathtt{send}\ A(\cdot)\to B$, listed in increasing local index $j$.
Define analogously the receive-event sequence
$
\mathsf{rcv}_{B\leftarrow A}(M)=(r_1,\ldots,r_n)
$
where $(r_i)_i$ are exactly the events $(B,j)\in E_M$ such that
$w_B[j]$ is of the form $\mathtt{recv}\ B(\cdot)\leftarrow A$, again listed in increasing local index.

The FIFO relation $\lhd_M\subseteq E_M\times E_M$ contains, for each channel
$A\to B$, the pairs $(s_i,r_i)$ for all $1\le i\le \min\{m,n\}$.
\end{definition}

\begin{definition}[Payload Matching \leanformalized{https://github.com/zippergen-io/zippergen-lean/blob/isola-camera-ready/isola/MSCAgents/PayloadMatching.lean}]
\label{def:payload-matching}
Let $\vec{x}$ be a sender payload tuple and let $\vec{y}$ be a
receiver payload tuple, whose components may be variables or constants.
We write $\mathsf{match}(\vec{x},\vec{y})$ if the following hold:
\begin{enumerate}
  \item $|\vec{x}|=|\vec{y}|$;
  \item for each index $i$:
  if $y_i$ is a constant, then $x_i$ is the same constant;
  if $y_i$ is a variable of type $\tau$, then $x_i$ is either a constant of
  type $\tau$ or a variable of type~$\tau$.
\end{enumerate}
Applied to a source-level $\mathtt{msg}\ A(\vec{x})\to B(\vec{y})$, this condition
is checked statically as part of well-typedness (see Definition~\ref{def:well-typed} below).
When $y_i$ is a constant, the distributed implementation can additionally verify
at runtime that the received value equals that constant, raising an error if not
(a lightweight consistency check requiring no extra specification).
The same conditions also describe runtime payload agreement for the control-tag
payloads introduced later by the projection.
\end{definition}

Payload matching isolates the only place where sender and receiver payloads are
allowed to differ syntactically: receivers may bind variables or require
distinguished constants. The raw workflow grammar itself does not enforce that
messages match in this sense, nor that actions and guards are locally
consistent. We therefore restrict attention to a well-typed fragment.

\begin{definition}[Well-Typed Global Workflows \leanformalized{https://github.com/zippergen-io/zippergen-lean/blob/isola-camera-ready/isola/MSCAgents/WellTyped.lean}]
\label{def:well-typed}
A global workflow $P$ is \emph{well typed} if all the following conditions
hold:
\begin{enumerate}
  \item every message subprogram
  $\mathtt{msg}\ A(\vec{x}) \to B(\vec{y})$
  satisfies $\mathsf{match}(\vec{x},\vec{y})$;
  \item every action subprogram
  $\ActF{A}{\vec{y}}{f(\vec{x})}$
  is locally type-correct, i.e., the invocation $f(\vec{x})$ is assumed
  to produce values compatible with the variables~$\vec{y}$;
  \item every control condition $c@B$ occurring in an $\mathtt{if}$ or
  $\mathtt{while}$ statement is well typed for lifeline $B$, i.e.,
  $c$ is a Boolean expression over variables available at $B$.
\end{enumerate}
Throughout the remainder of the paper, we consider only well-typed global workflows.
\end{definition}

In fact, only item~(1) is needed to ensure that the inductive MSC semantics below is
well defined. Items~(2)--(3) record additional static consistency
conditions expected of executable programs, but are not used explicitly in the
semantic constructions below.

We can now say which tuples of local words actually represent executions.
Intuitively, a valid MSC must have every receive matched by a send, matched
events must agree on their payloads, and the causal dependencies induced by
local order and communication must not form cycles.

\begin{definition}[MSC \leanformalized{https://github.com/zippergen-io/zippergen-lean/blob/isola-camera-ready/isola/MSCAgents/MSC.lean}]
\label{def:msc}
An \emph{MSC} is a tuple of local words
$M=(w_A)_{A\in\mathscr{L}}$ with $w_A\in\Sigma_A^*$ such that:
\begin{enumerate}
  \item every receive event in $w_A$ is matched in $\lhd_M$;
  \item 
  whenever a send event with payload tuple $\vec{x}$ is matched to a
  receive event with payload tuple $\vec{y}$, they satisfy
  $\mathsf{match}(\vec{x},\vec{y})$
  (in particular, receiver constants are matched by value equality);
  \item the transitive closure of the binary relation ${\lhd_M} \cup \{((A,i),(A,i+1))\mid A\in\mathscr{L},\,i\in\{1,\ldots,|w_A|-1\}\}$ is a strict partial order (i.e., acyclic).
\end{enumerate}
We call $M$ \emph{complete} if every send event is matched in $\lhd_M$.
\end{definition}

To compose executions sequentially, we use componentwise concatenation:
$M_1\circ M_2 \df (u_Av_A)_{A\in\mathscr{L}}$
for
$M_1=(u_A)_A,\ M_2=(v_A)_A$.

The concatenation of two MSCs is not necessarily an MSC.
However, we have the following closure facts:

\begin{lemma}[Concatenation with a Complete Prefix \leanproof{https://github.com/zippergen-io/zippergen-lean/blob/isola-camera-ready/isola/MSCAgents/MSC.lean\#L1159}]
\label{lem:concat-msc}
Let $M_1$ be a complete MSC and $M_2$ be an MSC. Then
$M_1\circ M_2$ is an MSC.
If moreover $M_2$ is complete, then $M_1\circ M_2$ is complete.
\end{lemma}

\begin{lemma}[Stripping a Complete Prefix \leanproof{https://github.com/zippergen-io/zippergen-lean/blob/isola-camera-ready/isola/MSCAgents/MSC.lean\#L1624}]
\label{lem:strip-complete-prefix}
Conversely, let $C=(u_A)_A$ be a complete MSC and let $N=(v_A)_A$ be a tuple of
local words.
If $C\circ N=(u_Av_A)_A$ is an MSC, then $N$ is an MSC.
If $C\circ N$ is a complete MSC, then $N$ is a complete MSC.
\end{lemma}

With these ingredients in place, the semantics can be defined
compositionally. The base cases are canonical one-step MSCs for the atomic
workflow constructs. Sequential composition uses~concatenation of MSCs, and
owned conditionals and loops contribute explicit choice events on the deciding
lifeline.

\begin{definition}[Canonical MSCs \leanformalized{https://github.com/zippergen-io/zippergen-lean/blob/isola-camera-ready/isola/MSCAgents/CanonicalMSC.lean}]
\label{def:base-msc}
Fix $A,B\in\mathscr{L}$, choice letter $\gamma\in\{\IfT{c}{B},\IfF{c}{B},\WhT{c}{B},\WhF{c}{B}\}$ for
some condition~$c$, tuples $\vec{x},\vec{y}$, and action $f\in\mathscr{F}$. We define the following canonical MSCs:
\begin{itemize}
\item$M_{\varepsilon}$ is the MSC with no events on any lifeline;
\item $M_{\gamma}^{B}
  =(w_X)_{X\in\mathscr{L}}$, where
  $
  w_B=\gamma$ and $w_X=\varepsilon~(X\neq B)
  $;
    \item $M_{\vec{x},\vec{y}}^{A,f}
  =(w_X)_{X\in\mathscr{L}}$, where
  $w_A=\ActF{A}{\vec{y}}{f(\vec{x})}$ and
  $w_X=\varepsilon~(X\neq A)$;

  \item $M_{\vec{x},\vec{y}}^{A\to B}
  =(w_X)_{X\in\mathscr{L}}$, where
  $w_A=\mathtt{send}\ A(\vec{x})\to B$,
  $w_B=\mathtt{recv}\ B(\vec{y})\leftarrow A$, and
  $w_X=\varepsilon~(X\notin\{A,B\})$.
\end{itemize}
\end{definition}

\begin{figure}[t]
\centering
\begin{tikzpicture}[scale=0.78,thick,
  action/.style={circle, fill=black, minimum size=5pt, inner sep=0pt},
  lbl/.style={font=\scriptsize},
  msg/.style={-{Triangle[length=4pt,width=4pt]}, line width=0.8pt},
  panel/.style={draw=black!30, rounded corners=2pt}]

\begin{scope}[xshift=0cm]
  \draw[panel] (-0.35,0.6) rectangle (1.75,-1.55);
  \node[lbl] at (0.7, 0.28) {$M_{\varepsilon}$};
  \node[lbl] at (0,0.16) {$A$};
  \node[lbl] at (1.4,0.16) {$B$};
  \draw (0,0) -- (0,-1.3);
  \draw (1.4,0) -- (1.4,-1.3);
\end{scope}

\begin{scope}[xshift=3cm]
  \draw[panel] (-0.35,0.6) rectangle (1.75,-1.55);
  \node[lbl] at (0.7, 0.28) {$M_{\gamma}^{B}$};
  \node[lbl] at (0,0.16) {$A$};
  \node[lbl] at (1.4,0.16) {$B$};
  \draw (0,0) -- (0,-1.3);
  \draw (1.4,0) -- (1.4,-1.3);
  \node[action] at (1.4,-0.65) {};
  \node[lbl] at (1.1,-0.65) {$\gamma$};
\end{scope}

\begin{scope}[xshift=6cm]
  \draw[panel] (-0.35,0.6) rectangle (2.05,-1.55);
  \node[lbl] at (0.85, 0.28) {$M_{\vec{x},\vec{y}}^{A,f}$};
  \node[lbl] at (0,0.16) {$A$};
  \node[lbl] at (1.7,0.16) {$B$};
  \draw (0,0) -- (0,-1.3);
  \draw (1.7,0) -- (1.7,-1.3);
  \node[action, label={[lbl]right:\itshape $f(\vec{x})/\vec{y}$}] at (0,-0.65) {};
\end{scope}

\begin{scope}[xshift=9.3cm]
  \draw[panel] (-0.35,0.6) rectangle (2.15,-1.55);
  \node[lbl] at (0.9, 0.28) {$M_{\vec{x},\vec{y}}^{A\to B}$};
  \node[lbl] at (0,0.16) {$A$};
  \node[lbl] at (1.8,0.16) {$B$};
  \draw (0,0) -- (0,-1.3);
  \draw (1.8,0) -- (1.8,-1.3);
  \draw[msg] (0,-0.65) -- (1.8,-0.65) node[midway, lbl, above] {$\vec{x},\vec{y}$};
\end{scope}

\end{tikzpicture}
\caption{Canonical one-step MSCs used in Definition~\ref{def:base-msc}. The schematic figure shows the empty MSC, a local choice event, a local action event, and a single message exchange.}
\label{fig:base-mscs}
\end{figure}

\begin{definition}[Inductive MSC Semantics \leanformalized{https://github.com/zippergen-io/zippergen-lean/blob/isola-camera-ready/isola/MSCAgents/InductiveSemantics.lean}]
\label{def:inductive-msc}
The semantics $\sem{P}$ of a well-typed global workflow $P$ is a set of MSCs.
It is defined by structural
recursion with the rules given in Table~\ref{tab:inductive-msc}.
\end{definition}

\begin{table}[t]
\centering
\caption{Inductive MSC semantics of global workflows.}
\label{tab:inductive-msc}
$
\begin{array}{rcl}
\sem{\varepsilon}
&\df& \{M_{\varepsilon}\} \\[0.3em]
\sem{\mathtt{msg}\ A(\vec{x}) \to B(\vec{y})}
&\df& \{M_{\vec{x},\vec{y}}^{A\to B}\} \\[0.3em]
\sem{\ActF{A}{\vec{y}}{f(\vec{x})}}
&\df& \{M_{\vec{x},\vec{y}}^{A,f}\} \\[0.3em]
\sem{P_1 ; P_2}
&\df& \{M_1\circ M_2 \mid M_1\in\sem{P_1},\ M_2\in\sem{P_2}\} \\[0.3em]
\sem{\mathtt{if}\ c@B\ \mathtt{then}\ P_{\top}\ \mathtt{else}\ P_{\bot}}
&\df&
\begin{array}[t]{@{}l@{}}
\{M_{\IfT{c}{B}}^{B}\circ M \mid M\in\sem{P_{\top}}\} \\
\cup\,\{M_{\IfF{c}{B}}^{B}\circ M \mid M\in\sem{P_{\bot}}\}
\end{array}
\\[1.5em]
\sem{\mathtt{while}\ c@B\ \mathtt{do}\ P_{\mathit{body}}\ \mathtt{exit}\ P_{\mathit{exit}}}
&\df&
\bigcup_{k\ge 0}\left\{
\begin{array}{@{}l@{}}
M_1^{\top}\circ\cdots\circ M_k^{\top}\circ M_{\mathit{exit}}^{\bot}
\ \mid\\[0.2em]
M_1,\ldots,M_k\in\sem{P_{\mathit{body}}},\\
M_{\mathit{exit}}\in\sem{P_{\mathit{exit}}}
\end{array}
\right\}
\end{array}
$\\[2ex]
where
$M_i^{\top}\df M_{\WhT{c}{B}}^{B}\circ M_i$ and
$M_{\mathit{exit}}^{\bot}\df M_{\WhF{c}{B}}^{B}\circ M_{\mathit{exit}}$.
\end{table}

One observes that for every well-typed global workflow $P$ and every $M\in\sem{P}$, the MSC $M$ is complete.
This follows by structural induction on $P$ from the semantic clauses.

\section{Distributed Agent Programs}
\label{sec:distributed}

A distributed agent program assigns one local program to each lifeline.
The local program at $A\in\mathscr L$ describes the steps visible at $A$:
sends and receives on $A$'s channels, local actions, and control steps arising
from branching decisions that $A$ either owns or observes via a received control signal.
\begin{definition}[Local Syntax \leanformalized{https://github.com/zippergen-io/zippergen-lean/blob/isola-camera-ready/isola/MSCAgents/LocalSyntax.lean}]
For each lifeline $A\in\mathscr L$, let $\mathsf{LocProg}_A$ be the set of local
programs at $A$, generated by the grammar
\begin{align*}
S::={}& \varepsilon \mid \mathtt{send}\ A(\vec{x})\to B\mid
\mathtt{recv}\ A(\vec{y})\leftarrow B\mid\ActF{A}{\vec{y}}{f(\vec{x})}\mid
S_1;S_2\mid\\
&\mathtt{if}\ c@A\ \mathtt{then}\ S_{\top}\ \mathtt{else}\ S_{\bot}\mid
\mathtt{if}\ A(\vec{y})\leftarrow B\ \mathtt{then}\ S_{\top}\ \mathtt{else}\ S_{\bot}\mid
\\
&\mathtt{while}\ c@A\ \mathtt{do}\ S_{\mathit{body}}\ \mathtt{exit}\ S_{\mathit{exit}}
\mid \mathtt{while}\ A(\vec{y})\leftarrow B\ \mathtt{do}\ S_{\mathit{body}}\ \mathtt{exit}\ S_{\mathit{exit}}
\end{align*}
where $c$ ranges over source guards.
Here, each occurrence of $B$ ranges over lifelines in $\mathscr L\setminus\{A\}$.
In $\mathtt{send}\ A(\vec{x})\to B$ and $\mathtt{recv}\ A(\vec{y})\leftarrow B$,
$\vec{x}$ and $\vec{y}$ range over tuples of constants and variables.
In action statements, $\vec{y}$ must be variables (assignment targets).
For constructs $\mathtt{if}/\mathtt{while}\ A(\vec{y})\leftarrow B$, $\vec{y}$ is a
non-empty tuple of constants and variables whose first component $y_1$ is a Boolean variable.
\end{definition}

\begin{definition}[Distributed Programs \leanformalized{https://github.com/zippergen-io/zippergen-lean/blob/isola-camera-ready/isola/MSCAgents/LocalSyntax.lean}]
A distributed program is a tuple
$\mathcal D=(S_A)_{A\in\mathscr L}$ with $S_A\in\mathsf{LocProg}_A$ for each $A$.
\end{definition}

The message and action constructs are the local counterparts of the corresponding
workflow constructs, as are the owner-side conditional control-flow constructs $\mathtt{if}\ c@A$ and $\mathtt{while}\ c@A$: they still represent source-level local control owned by $A$.
The receive-guarded forms
$\mathtt{if}/\mathtt{while}\ A(\vec{y})\leftarrow B$ are different: they do not
evaluate a guard locally, but instead react to a control decision that has
already been made by $B$ and communicated explicitly.

The semantics of a local program is a language of local traces, one word over
the alphabet $\Sigma_A$ from Definition~\ref{def:alphabets} for each possible
execution of lifeline $A$. The semantics of a distributed program then
contains the tuples of
such local traces that together form an MSC. In this way, the distributed
semantics reuses the MSC notion from Section~\ref{sec:syntax}, but starts from
already projected local behavior rather than from a global workflow directly.

For $S \in \mathsf{LocProg}_A$, we define
$\ltr{A}{S}\subseteq\Sigma_A^*$, the set of local traces of $S$,
by induction. The receive-guarded control constructs
$\mathtt{if}/\mathtt{while}\ A(\vec{y})\leftarrow B$ consume control-tagged
messages from $B$ and branch on the first component ($\top$ or $\bot$).
The formal definition is given in Table~\ref{tab:local-semantics};
there, we write
$\vec{y}[\nu/y_1]$ for the tuple obtained by replacing $y_1$ by $\nu$.

To reason about partial executions and
to formalize deadlock-freeness, we order local words and MSC tuples by the
componentwise prefix relation.
For words $u,u'\in\Sigma_A^*$, write $u\preceq u'$ if
there is $v\in\Sigma_A^*$ such that $u'=uv$. We write $u\prec u'$ if
$u\preceq u'$ and $u\neq u'$.

\begin{table}[t]
\centering
\caption{Local trace semantics for $\ltr{A}{S}$.}
\label{tab:local-semantics}
$
\begin{array}{rcl}
\ltr{A}{\varepsilon} &\df& \{\varepsilon\}\\
\ltr{A}{\mathtt{send}\ A(\vec{x})\to B}
&\df& \{\mathtt{send}\ A(\vec{x})\to B\}\\
\ltr{A}{\mathtt{recv}\ A(\vec{y})\leftarrow B}
&\df& \{\mathtt{recv}\ A(\vec{y})\leftarrow B\}\\
\ltr{A}{\ActF{A}{\vec{y}}{f(\vec{x})}}
&\df& \{\ActF{A}{\vec{y}}{f(\vec{x})}\}\\
\ltr{A}{S_1;S_2} &\df& \{u_1u_2\mid u_1\in\ltr{A}{S_1},\ u_2\in\ltr{A}{S_2}\}\\[0.3em]
\ltr{A}{\mathtt{if}\ c@A\ \mathtt{then}\ S_{\top}\ \mathtt{else}\ S_{\bot}}
&\df& \left(\begin{array}{rl}
& \{\IfT{c}{A}\cdot u \mid u\in\ltr{A}{S_{\top}}\}\\
\cup &\{\IfF{c}{A}\cdot u \mid u\in\ltr{A}{S_{\bot}}\}
\end{array}\right)\\[1em]
\ltr{A}{\mathtt{if}\ A(\vec{y})\leftarrow B\ \mathtt{then}\ S_{\top}\ \mathtt{else}\ S_{\bot}}
&\df& \left(\begin{array}{rl}
& \{\mathtt{recv}\ A(\vec{y}[\top/y_1])\leftarrow B\cdot u \mid u\in\ltr{A}{S_{\top}}\}\\
\cup &\{\mathtt{recv}\ A(\vec{y}[\bot/y_1])\leftarrow B\cdot u \mid u\in\ltr{A}{S_{\bot}}\}
\end{array}\right)
\end{array}
$
$
\begin{aligned}
\ltr{A}{\mathtt{while}\ c@A\ \mathtt{do}\ S_{\mathit{body}}\ \mathtt{exit}\ S_{\mathit{exit}}}
&\df\\
&\bigcup_{k\ge 0}\left\{
\begin{array}{l}
\WhT{c}{A}\cdot u_1\cdot\\
\vdots\\
\WhT{c}{A}\cdot u_k\cdot\\
\WhF{c}{A}\cdot v
\end{array}
\;\middle|\;
\begin{array}{l}
u_1,\ldots,u_k\in\ltr{A}{S_{\mathit{body}}},\\
v\in\ltr{A}{S_{\mathit{exit}}}
\end{array}
\right\}
\end{aligned}
$
$
\begin{aligned}
\ltr{A}{\mathtt{while}\ A(\vec{y})\leftarrow B\ \mathtt{do}\ S_{\mathit{body}}\ \mathtt{exit}\ S_{\mathit{exit}}}
&\df\\
&\hspace*{-6.8em}\bigcup_{k\ge 0}\left\{
\begin{array}{l}
\mathtt{recv}\ A(\vec{y}[\top/y_1])\leftarrow B\cdot u_1\cdot\\
\vdots\\
\mathtt{recv}\ A(\vec{y}[\top/y_1])\leftarrow B\cdot u_k\cdot\\
\mathtt{recv}\ A(\vec{y}[\bot/y_1])\leftarrow B\cdot v
\end{array}
\;\middle|\;
\begin{array}{l}
u_1,\ldots,u_k\in\ltr{A}{S_{\mathit{body}}},\\
v\in\ltr{A}{S_{\mathit{exit}}}
\end{array}
\right\}
\end{aligned}
$
\end{table}

A language \(L\subseteq\Sigma_A^*\) is called \emph{prefix-free} if no two
distinct words in \(L\) are prefix-comparable.

\begin{lemma}[Complete Local Traces Are Prefix-Free \leanproof{https://github.com/zippergen-io/zippergen-lean/blob/isola-camera-ready/isola/MSCAgents/LocalSemantics.lean\#L1118}]
\label{lem:local-prefix-free}
For every local program $S\in\mathsf{LocProg}_A$, the language
\(\ltr{A}{S}\) is prefix-free.
\end{lemma}

Now that we have defined the local trace languages, the semantics of a distributed
program is obtained by considering the tuples with one local trace for each lifeline, and keeping
exactly those tuples that satisfy the MSC conditions from
Definition~\ref{def:msc}. We define two semantics for distributed programs: one for complete executions and one for partial executions, where each local trace may be a prefix of a complete one.

Given a distributed program
$\mathcal D=(S_A)_{A\in\mathscr L}$, its complete MSC semantics is
\[
\sem{\mathcal D}\df
\left\{M=(w_A)_{A\in\mathscr{L}}\ \middle|\
\begin{array}{l}
w_A\in\ltr{A}{S_A}\ \text{for all }A,\\
M\ \text{is a complete MSC}
\end{array}
\right\}.
\]
For $S \in \mathsf{LocProg}_A$,
we consider the prefix closure $\lpref{A}{S}\subseteq\Sigma_A^*$, defined as
$\lpref{A}{S}\df\{u\in\Sigma_A^*\mid \exists v\in\ltr{A}{S}\colon u\preceq v\}$.
The distributed prefix semantics is
\[
\sempref{\mathcal D}\df
\left\{M=(w_A)_{A\in\mathscr{L}}\ \middle|\
\begin{array}{l}
w_A\in\lpref{A}{S_A}\ \text{for all }A,\\
M\ \text{is an MSC}
\end{array}
\right\}.
\]
Clearly, $\sem{\mathcal D}\subseteq\sempref{\mathcal D}$.
In particular, $\sem{\mathcal D}$ contains only complete MSCs (hence MSCs) by
Definition~\ref{def:msc}.
Note that MSC-ness is defined directly on tuples of local words, independently of any global workflow.

For tuples $M=(w_A)_{A\in\mathscr L}$ and $M'=(w'_A)_{A\in\mathscr L}$, define
$M\preceq M'$ if, for all $A\in\mathscr L$, we have $w_A\preceq w'_A$.

\begin{definition}[Deadlock-Freeness \leanformalized{https://github.com/zippergen-io/zippergen-lean/blob/isola-camera-ready/isola/MSCAgents/DeadlockFreeness.lean}]
\label{def:deadlock-prefix}
A distributed program $\mathcal D$ is \emph{deadlock-free} if,
for all $M\in\sempref{\mathcal D}$, there is $M'\in\sem{\mathcal D}$ such that $M\preceq M'$.
\end{definition}

For arbitrary distributed programs, such a prefix-extension property can be
misleading: a prefix may not determine which branch choices have already been
made, so an extension might complete a different run.  Our projected programs
avoid this ambiguity.  Branch decisions are recorded locally, as owner choice
events or tagged control receives, and complete local trace languages are
prefix-free (Lemma~\ref{lem:local-prefix-free}).  Thus a prefix fixes all branch
choices that have occurred. Choices not yet visible have not yet been made.
Therefore Definition~\ref{def:deadlock-prefix} expresses the intended no-stuck
property for projected programs.

\section{From Global to Distributed Programs}
\label{sec:implementation}

This section explains how a global workflow is turned into a distributed
program. We first define the syntax-directed projection itself, and then state
its correctness with respect to the MSC semantics.

\subsection{Projection}
\label{sec:projection}

We now pass from global workflows to executable local programs.
The projection is syntax-directed: each workflow construct is translated
locally for every lifeline, and the resulting local programs together form the
distributed program that will be executed. Two ingredients are needed.
First, for owned control constructs, the deciding lifeline must know which
other lifelines may participate in the continuation. Second, the projected
implementation uses explicit control broadcasts that are not present in the
global semantics, so we compare projected executions to global ones only after
erasing these auxiliary events.

To determine which lifelines may need to observe a control decision, we use a
simple structural participation analysis.
\begin{definition}[Structural Participation Sets \leanformalized{https://github.com/zippergen-io/zippergen-lean/blob/isola-camera-ready/isola/MSCAgents/ParticipationSets.lean}]
\label{def:participation-sets}
For a global workflow $P$, the participation set $\LS{P}$ is defined
inductively as follows:
\begin{align*}
\LS{\mathtt{msg}\ A(\vec{x})\to B(\vec{y})} &= \{A,B\},\qquad
\LS{\varepsilon} = \varnothing,\\
\LS{\ActF{A}{\vec{y}}{f(\vec{x})}} &= \{A\},\qquad
\LS{P_1;P_2} = \LS{P_1}\cup\LS{P_2},\\
\LS{\mathtt{if}\ c@B\ \mathtt{then}\ P_{\top}\ \mathtt{else}\ P_{\bot}}
&= \{B\}\cup\LS{P_{\top}}\cup\LS{P_{\bot}},\\
\LS{\mathtt{while}\ c@B\ \mathtt{do}\ P_{\mathit{body}}\ \mathtt{exit}\ P_{\mathit{exit}}}
&= \{B\}\cup\LS{P_{\mathit{body}}}\cup\LS{P_{\mathit{exit}}}.
\end{align*}
\end{definition}

Using these sets, projection turns each owned decision into an explicit
broadcast from the decider to the lifelines that may continue afterwards.
Fix a total order $\sqsubset$ on lifelines. For $R\subseteq\mathscr{L} \setminus \{B\}$, we write
$
\prod\nolimits_{A\in R}^{\sqsubset}
\mathtt{send}\ B(\nu,\kappa_{\mathsf{ctrl}}^P)\to A \in \mathsf{LocProg}_B
$
for the sequence of decision sends in $\sqsubset$-order, where
$\nu\in\{\top,\bot\}$ is the decision value and $P$ is the construct being projected;
if $R=\emptyset$, the product is the empty program~$\varepsilon$.
Each control-flow construct $P$ uses a dedicated control tag $\kappa_{\mathsf{ctrl}}^P$
(indexed by that construct). Control tags have a dedicated type that is unavailable
to user payloads: no source-level variable or constant has this type. Consequently, an
ordinary user receive can neither declare a variable of the tag type nor supply the
reserved tag as a constant, so by payload matching (Definition~\ref{def:payload-matching})
a user receive can never match a control send, and conversely.
Two constructs that share a tag must be syntactically equal and therefore cannot be nested
inside each other.
Thus, control messages of nested constructs always carry distinct tags and
are syntactically distinguishable.

These control broadcasts are implementation details: they are needed to
realize owned control in distributed code, but they should not appear in the
global workflow semantics. We therefore compare projected executions to global
executions only after erasing them.
For each lifeline $A$, let
$\mathsf{erase}_A:\Sigma_A^*\to\Sigma_A^*$ delete all
decision-broadcast send/receive letters introduced by projection for control
constructs (the $\mathtt{send}\ B(\top/\bot,\kappa_{\mathsf{ctrl}}^P)\to A$ and
$\mathtt{recv}\ A(\top/\bot,\kappa_{\mathsf{ctrl}}^P)\leftarrow B$ events, for any control construct $P$,
used only for control decisions).
For an MSC $M=(w_A)_A$, define
$
\mathsf{erase}(M)\df(\mathsf{erase}_A(w_A))_{A\in\mathscr{L}}.
$

Because the control-tag type is unavailable to user payloads, each matched pair of a
send and its receive consists either of two control events or of two user events, so
erasure removes complete control pairs without disturbing user communication.

\begin{lemma}[Erasure Preserves MSC and Completeness \leanproof{https://github.com/zippergen-io/zippergen-lean/blob/isola-camera-ready/isola/MSCAgents/Erasure.lean\#L1585}]
\label{lem:erase-preserves}
If $M$ is an MSC, then $\mathsf{erase}(M)$ is an MSC.
If $M$ is a complete MSC, then $\mathsf{erase}(M)$ is complete.
\end{lemma}

With these conventions fixed, the projection rules themselves are
straightforward. Atomic statements are projected pointwise to the relevant
endpoints or owner lifeline.
Sequential composition is projected componentwise.
The only interesting case is owned control: the owner keeps its source-level
\texttt{if}/\texttt{while}, emits explicit control messages, and each recipient
branches by receiving a reserved control token from the owner.
For each lifeline $A\in\mathscr L$, we define
$\pi_A\colon\Work\to\mathsf{LocProg}_A$.
The projection rules are given in
Tables~\ref{tab:projection-base},
\ref{tab:projection-if}, and
\ref{tab:projection-while}, where we let
$\mathcal{R}_{\mathit{if}} \df (\LS{P_{\top}}\cup \LS{P_{\bot}})\setminus\{B\}$
and
$\mathcal{R}_{\mathit{while}} \df (\LS{P_{\mathit{body}}}\cup \LS{P_{\mathit{exit}}})\setminus\{B\}$.
For recipient-side control constructs, $(z,\kappa_{\mathsf{ctrl}}^P)$
is the received payload, where $z$ is a Boolean variable not occurring in $P$ and
$\kappa_{\mathsf{ctrl}}^P$ is the control tag of the enclosing construct $P$.

For a global workflow $P$, we write
$\mathcal D_P \df (\proj{A}{P})_{A\in\mathscr L}$
for its projected distributed program.

Let \(|P|\) be the number of syntactic nodes of \(P\) and \(n = |\mathscr L|\).
Each global construct contributes at most one local node per lifeline, giving at most \(n|P|\) nodes overall.
Moreover, each control construct adds at most two owner-side broadcast sequences, each of length at most \(n-1\).
Hence the total size of all projections is at most \(n|P| + 2(n-1)|P| = O(n|P|)\), and so each individual projection \(\proj{X}{P}\) is also \(O(n|P|)\).

\begin{table}[t]
\centering
\caption{Projection rules for atomic and sequential constructs.}
\label{tab:projection-base}
\small
$
\begin{array}{rcl}
\proj{A}{\varepsilon}
&\df& \varepsilon
\qquad
\proj{A}{P_1;P_2}
\df \proj{A}{P_1};\proj{A}{P_2}
\\[0.5ex]
\proj{A}{\mathtt{msg}\ X(\vec{x})\to Y(\vec{y})}
&\df& \begin{cases}
\mathtt{send}\ X(\vec{x})\to Y & \text{if }A=X\\
\mathtt{recv}\ Y(\vec{y})\leftarrow X & \text{if }A=Y\\
\varepsilon & \text{otherwise}
\end{cases}\\[5ex]
\proj{A}{\ActF{X}{\vec{y}}{f(\vec{x})}}
&\df& \begin{cases}
\ActF{X}{\vec{y}}{f(\vec{x})} & \text{if }A=X\\
\varepsilon & \text{otherwise}
\end{cases}
\end{array}
$
\end{table}

\begin{table}[t]
\centering
\caption{Projection rule for $P=\mathtt{if}\ c@B\ \mathtt{then}\ P_{\top}\ \mathtt{else}\ P_{\bot}$.}
\label{tab:projection-if}
$
\proj{A}{\mathtt{if}\ c@B\ \mathtt{then}\ P_{\top}\ \mathtt{else}\ P_{\bot}}
\df
\begin{cases}
\left(
\begin{array}{@{}l@{}}
\mathtt{if}\ c@B\ \mathtt{then}\\
\hspace*{1.5em}\prod\nolimits_{C\in\mathcal{R}_{\mathit{if}}}^{\sqsubset}
\mathtt{send}\ B(\top,\kappa_{\mathsf{ctrl}}^P)\to C;\\
\hspace*{1.5em}\proj{B}{P_{\top}}\\
\mathtt{else}\\
\hspace*{1.5em}\prod\nolimits_{C\in\mathcal{R}_{\mathit{if}}}^{\sqsubset}
\mathtt{send}\ B(\bot,\kappa_{\mathsf{ctrl}}^P)\to C;\\
\hspace*{1.5em}\proj{B}{P_{\bot}}
\end{array}
\right)
& \text{if }A=B\\[10ex]
\left(
\begin{array}{@{}l@{}}
\mathtt{if}\ A(z,\kappa_{\mathsf{ctrl}}^P)\leftarrow B\ \mathtt{then}\\
\hspace*{1.5em}\proj{A}{P_{\top}}\\
\mathtt{else}\\
\hspace*{1.5em}\proj{A}{P_{\bot}}
\end{array}
\right)
& \text{if }A\in\mathcal{R}_{\mathit{if}}\\[6ex]
\varepsilon & \text{otherwise}
\end{cases}
$
\end{table}

\begin{table}[t]
\centering
\caption{Projection rule for $P=\mathtt{while}\ c@B\ \mathtt{do}\ P_{\mathit{body}}\ \mathtt{exit}\ P_{\mathit{exit}}$.}
\label{tab:projection-while}
$
\begin{gathered}
\llap{$\proj{A}{\mathtt{while}\ c@B\ \mathtt{do}\ P_{\mathit{body}}\ \mathtt{exit}\ P_{\mathit{exit}}}
\df$}\\
\begin{cases}
\left(
\begin{array}{@{}l@{}}
\mathtt{while}\ c@B\ \mathtt{do}\\
\hspace*{1.5em}\prod\nolimits_{C\in\mathcal{R}_{\mathit{while}}}^{\sqsubset}
\mathtt{send}\ B(\top,\kappa_{\mathsf{ctrl}}^P)\to C;\\
\hspace*{1.5em}\proj{B}{P_{\mathit{body}}}\\
\mathtt{exit}\\
\hspace*{1.5em}\prod\nolimits_{C\in\mathcal{R}_{\mathit{while}}}^{\sqsubset}
\mathtt{send}\ B(\bot,\kappa_{\mathsf{ctrl}}^P)\to C;\\
\hspace*{1.5em}\proj{B}{P_{\mathit{exit}}}
\end{array}
\right)
& \text{if }A=B\\[10ex]
\left(
\begin{array}{@{}l@{}}
\mathtt{while}\ A(z,\kappa_{\mathsf{ctrl}}^P)\leftarrow B\ \mathtt{do}\\
\hspace*{1.5em}\proj{A}{P_{\mathit{body}}}\\
\mathtt{exit}\\
\hspace*{1.5em}\proj{A}{P_{\mathit{exit}}}
\end{array}
\right)
& \text{if }A\in\mathcal{R}_{\mathit{while}}\\[6ex]
\varepsilon & \text{otherwise}
\end{cases}
\end{gathered}
$
\end{table}

Two design choices for local control keep the formal layer compact.
First, we keep dedicated recipient-side constructs
\(\mathtt{if}/\mathtt{while}\ A(z,\kappa_{\mathsf{ctrl}}^P)\leftarrow B\) rather
than rewriting them into plain \(\mathtt{recv}\) followed by ordinary
\(\mathtt{if}/\mathtt{while}\): this preserves an explicit receive-to-branch
coupling without introducing guard evaluation machinery.
Second, control broadcasts use constant values \(\top,\bot\) tagged with the
per-construct tag \(\kappa_{\mathsf{ctrl}}^P\). Indexing the tag by the construct $P$
ensures that control messages of distinct nested constructs are syntactically distinguishable,
which simplifies the decision-prefix attribution argument in the proof
of Lemma~\ref{lem:uniform-zipper}; the machine-checked proof does not
require this distinctness.
The tag also cleanly separates control traffic from user messages and identifies
exactly the events removed by the erasure map. In particular, payload
matching ensures that a receive expecting \((\top,\kappa_{\mathsf{ctrl}}^P)\)
cannot match a send carrying \((\bot,\kappa_{\mathsf{ctrl}}^P)\), and conversely.

\subsection{Correctness}

We can now state the main correctness result of the projection. The theorem establishes that global and projected         
  executions agree up to control messages: every global execution is realized by a distributed one after inserting the
  necessary control broadcasts, and every complete distributed execution yields a global one once control messages are       
  removed.

\begin{theorem}[Inductive Correctness \leanproof{https://github.com/zippergen-io/zippergen-lean/blob/isola-camera-ready/isola/MSCAgents/Correctness.lean\#L2055}]
\label{thm:correctness}
Let $P$ be a well-typed global workflow. Then the following hold:
\begin{enumerate}
  \item For every $M\in\sem{P}$, there exists
  $\enriched{M}\in\sem{\mathcal{D}_P}$ such that $\mathsf{erase}(\enriched{M})=M$.
  \item For every $\enriched{M}\in\sem{\mathcal{D}_P}$,
  $\mathsf{erase}(\enriched{M})\in\sem{P}$.
\end{enumerate}
\end{theorem}

In the machine-checked statements of both items, source-level control-payload
separation appears as an additional hypothesis: source-level message and action
payloads never inhabit the dedicated control-tag type. Apart from this hypothesis,
the theorem assumes only well-typedness. The deadlock-freedom corollary
(Corollary~\ref{cor:deadlock-free} below) does not require this additional
source-level separation hypothesis.

Our semantics ranges only over finite, complete MSCs and their finite prefixes.
Accordingly, Theorem~\ref{thm:correctness} is a partial-correctness statement
about terminating workflow executions: whenever a global or projected run
reaches a complete finite execution, the erasure-based correspondence holds.
In particular, every \texttt{while} loop is assumed to terminate,
i.e., the owner's condition eventually becomes false.
The theorem is not a separate claim that arbitrary subsymbolic actions
terminate, nor does it address infinite reactive behavior. This restriction is
appropriate for the protocol workflows considered here.
Extending the framework
to reactive settings with infinite executions would require a corresponding
semantic extension.

Both items are proved by structural induction on $P$. Completeness (item~1) is
straightforward: each case follows directly by induction. Soundness (item~2) is harder in the sequential composition case:
given a distributed execution of $P_1;P_2$, one must show that the
$P_1$- and $P_2$-parts of each local trace together form valid MSCs. The Zipper
lemma resolves exactly this difficulty. Deadlock-freeness
(Corollary~\ref{cor:deadlock-free}) then follows directly: applying the Zipper
lemma with empty suffix extends any prefix execution to a complete one.

\begin{definition}[Zipper postcondition \leanformalized{https://github.com/zippergen-io/zippergen-lean/blob/isola-camera-ready/isola/MSCAgents/ZipperPost.lean\#L120}]
\label{def:zip-post}
Let \(P\) be a well-typed global workflow, and let
\(U=(u_X)_{X\in\mathscr L}\), \(\bar U=(\bar u_X)_{X\in\mathscr L}\), and
\(V=(v_X)_{X\in\mathscr L}\) be tuples of local words.
We write
$
\mathsf{ZipPost}_{P}(U,\bar U,V)
$
if the following hold:
\begin{enumerate}
\item for every \(X\in\mathscr L\),
$u_X\preceq \bar u_X$,
$\bar u_X\in\ltr{X}{\proj{X}{P}}$,
and
$(v_X\neq\varepsilon ~\Longrightarrow~ \bar u_X=u_X)$;
\item \(\bar U\) is a complete MSC;
\item \((u_Xv_X)_{X\in\mathscr L}\) is an MSC;
\item \((\bar u_Xv_X)_{X\in\mathscr L}\) is an MSC.
\end{enumerate}
\end{definition}

The key result is that the Zipper postcondition is always satisfiable.

\begin{lemma}[Zipper lemma \leanproof{https://github.com/zippergen-io/zippergen-lean/blob/isola-camera-ready/isola/MSCAgents/ZipperLemma.lean\#L6859}]
\label{lem:uniform-zipper}
Let $P$ be a well-typed global workflow.
Let $U=(u_X)_{X\in\mathscr L}$ and $V=(v_X)_{X\in\mathscr L}$ be tuples of local
words such that:
\begin{itemize}
\item for every lifeline $X\in\mathscr L$,
$u_X\in\lpref{X}{\proj{X}{P}}$;

\item for every lifeline $X\in\mathscr L$,
$
v_X\neq\varepsilon ~\Longrightarrow~ u_X\in\ltr{X}{\proj{X}{P}}
$.
\end{itemize}
Define $M=(u_Xv_X)_{X\in\mathscr L}$.
If $M$ is an MSC, then there exists a tuple
$\bar U=(\bar u_X)_{X\in\mathscr L}$ such that
\[
\mathsf{ZipPost}_{P}(U,\bar U,V).
\]
\end{lemma}

Intuitively, the lemma addresses the key difficulty of asynchronous sequential
composition: in an execution of $P_1;P_2$, different lifelines may cross the
$P_1/P_2$-boundary at different times, leaving some still inside $P_1$ while
others have already begun $P_2$. The lemma guarantees that $P_1$ can always be
completed consistently with whatever $P_2$-work has started. This is the
situation illustrated by \texttt{reviewed\_execution}.

Three auxiliary consequences of \(\mathsf{ZipPost}_P(U,\bar U,V)\) will be used
below. First, the prefix tuple \(U\) is itself an MSC. Second, if each
component \(u_X\) is already a complete local trace of \(\proj{X}{P}\), then
\(U\) is in fact a complete MSC. Third, \(V\) is an MSC: conditions~2 and~4 together with
Lemma~\ref{lem:strip-complete-prefix} give this directly, so it need not be
imposed separately in the definition.

\begin{corollary}[Deadlock-freeness \leanproof{https://github.com/zippergen-io/zippergen-lean/blob/isola-camera-ready/isola/MSCAgents/DeadlockFreeness.lean\#L47}]
\label{cor:deadlock-free}
For every well-typed global workflow $P$, the distributed program
\(\mathcal D_P\) is deadlock-free in the sense of
Definition~\ref{def:deadlock-prefix}.
\end{corollary}

Corollary~\ref{cor:deadlock-free} is a semantic prefix-extension property:
every prefix MSC of the projected distributed program can be extended to a
complete MSC. This is the no-stuck property formalized in
Definition~\ref{def:deadlock-prefix}. It is not, by itself, a claim about
fairness or lock-freedom of every possible runtime scheduler.

\subsection{The ZipperGen System}
\label{sec:zippergen}

ZipperGen, available at \url{https://zippergen.io}, is a Python implementation of the framework described in this paper, whose internal structure follows the formal development step by step.
Workflows are written as ordinary Python functions and are decorated with
\texttt{@workflow}. The decorator rewrites the function body into an immutable
internal representation (IR) whose nodes mirror the formal grammar of Section~\ref{sec:syntax}: typed
inputs become owned initial variables, message and action statements become
message and action nodes, and Python \texttt{if}/\texttt{while} statements with
conditions of the form \texttt{expr~@~Lifeline} become owned control nodes.
LLM, pure, effect, and human actions are declared separately with \texttt{@llm},
\texttt{@pure}, \texttt{@effect}, and \texttt{@human}, attaching prompts,
implementations, effects, or human-input specs while the workflow itself remains a
pure coordination object.

Execution follows the projected semantics closely. ZipperGen projects the global
workflow, spawns one thread per lifeline, and connects lifelines by FIFO queues,
one per directed channel. Control broadcasts use per-construct tags $\kappa_{\mathsf{ctrl}}^P$,
so they remain disjoint from user traffic and from each other across nested constructs. The
LLM backend is pluggable, with a mock backend for testing without API access.

For deployment, ZipperGen runs workflows through a persistent SQLite store that
records messages, control decisions, action results, human tasks, trace events,
and the final result. Restarting with the same store replays the committed history
and continues from it. Human actions become durable tasks, completed through
command-line commands or notification adapters such as Telegram, and inspected with
\texttt{status}, \texttt{tasks}, and \texttt{trace}. The implementation includes a
local-approval example that waits for approval and replays without re-asking.

\paragraph{Runtime Planning.}
\label{sec:planner}

The framework as presented treats workflows as static specifications.
ZipperGen also supports a \texttt{@planner} decorator for cases where
the coordination structure is not known in advance.
A planner action is declared like any other action and can be used directly inside a \texttt{@workflow}:

\begin{lstlisting}[float=!tb,language=Python,caption={A planner action in actual Python syntax.}]
@planner(
    description="A workflow planner for professional writing tasks.",
    lifelines=["Worker1", "Worker2"],
    allow=["llm", "if"],
    instructions="Use Worker1 to draft, Worker2 to assess quality; "
                  "route back to Worker1 for revision if needed.",
)
def write_document(request: str, job_desc: str, cv_sketch: str) -> str: ...
\end{lstlisting}

At runtime, ZipperGen builds a hidden system prompt from all of the following:
the \texttt{description}, an auto-generated summary of the available workers, DSL rules,
and any \texttt{allow} extensions.
An optional \texttt{actions} parameter supplies a pre-defined vocabulary.
Since \texttt{@llm}, \texttt{@pure}, \texttt{@effect}, and \texttt{@human}
actions present the same interface of typed inputs and a typed output, any of
them may appear in the vocabulary.
The planner receives their signatures and may call them directly in the generated workflow.
The \texttt{allow} parameter controls what else the generated workflow may use:
\texttt{"pure"} and \texttt{"llm"} permit the LLM to define new actions beyond
the supplied vocabulary, while \texttt{"if"} and \texttt{"while"} enable
conditional branching and loops.
By default only linear workflows over the supplied action vocabulary are permitted.
The optional \texttt{instructions} parameter lets the user guide how workers are
assigned to tasks.
Without it the runtime encourages the planner to use as many workers as
reasonable.

The designated LLM also receives the actual values of all input variables, so it
sees what it is working with rather than just variable names, and generates a
complete workflow in the ZipperGen DSL.
ZipperGen validates the generated workflow, projects it, and executes it.
Worker lifelines are declared by name in \texttt{@planner} and created implicitly
for the generated workflow's isolated scope, with their own threads and channels.
Worker names must be distinct from the calling lifeline.
Under these conditions the same projection and deadlock-freedom guarantees apply as for handwritten workflows.

Runtime generation does introduce uncertainty: the generated workflow is
LLM-produced and may not behave as intended.
Three properties of the framework help manage this.
First, each generated workflow has its own lifeline scope: it cannot access the
outer workflow's lifelines or local variables, and communicates back to the outer
workflow only through its declared return value. Declared effect actions may
nevertheless affect external state.
Second, structural invariants are checked before execution.
In particular, the generated workflow must end with \texttt{return}~\textit{var}~\texttt{@}~\textit{Lifeline},
where \textit{var} is in scope on \textit{Lifeline} at that point.
Third, generated workflows are surfaced in the runtime trace and logs,
so the produced coordination code can be inspected after validation and
execution.

\paragraph{Lightweight planner experiment.}
We use the arithmetic planner example as a small check of runtime planning.
The planner receives an arithmetic expression and must generate a ZipperGen
workflow that distributes the computation over three calculator lifelines and
calls deterministic calculator actions. The expression family is
\((2 - 4) * (2 + 3) + (3 / (3 - x))\), with \(x\in\{2,3\}\).
For \(x=3\), the final division has denominator \(0\). The expected global
result is then defined to be \(0\), so the planner has to generate a conditional
workflow rather than simply call division. We call a generated workflow \emph{accepted} if it parses and passes all static
checks (typing and dataflow validation). An attempt is retried whenever parsing or
static validation fails, and the retry prompt includes the validation error and the
rejected workflow. Once a workflow is accepted, it is executed without
further regeneration. For each model, we run each expression five times and allow
up to 8 attempts per run. We report first-attempt acceptance, acceptance within
the retry budget, the capped mean number of attempts, and whether the accepted
workflows execute locally and return the same numeric result as a direct evaluator
with the same division-by-zero convention.
The zero-shot prompt gives the workflow syntax, the three calculator lifelines,
the typed action vocabulary, and the convention for division by zero, but not a
target workflow or worked example. Unaccepted runs count as 8 attempts in the
capped mean. This is not a benchmark of arithmetic ability, but a check of whether
models can produce workflows that enter ZipperGen's validation, projection, and
execution pipeline.

Correct workflows for \(x=3\)
compute the denominator, test it with \texttt{is\_zero}, and return \(0\) for
the whole expression. For the incorrect accepted workflows, inspecting the
generated code showed that they applied the zero convention only to the division
subexpression, yielding \(-10\) instead of \(0\) for \(x=3\). Both local
Qwen models and GPT-4o mini failed to produce an accepted workflow within
the retry budget: their generated
workflows were rejected by the syntax and dataflow checks.
The failures separate coordination from task semantics: invalid workflows are
rejected before execution, while executable workflows may still implement the
wrong arithmetic convention. The latter affects \emph{Correct}, but not the
projection guarantees.
The expression set, planner prompt, retry policy, and evaluation script are included
with the implementation. The numbers reflect this particular setup, not a general
ranking of the models.

\begin{table}[!t]
\centering
\caption{Small runtime-planning check using the arithmetic planner.
Each model is run 5 times on each expression, with up to 8 attempts per run.
All accepted workflows are projected and executed locally.}
\label{tab:planner-eval}
\small
\begingroup
\newcommand{\plannerhead}[1]{\hspace{0.4em}#1\hspace{0.4em}}
\begin{tabular*}{\textwidth}{@{\extracolsep{\fill}}*{7}{c}@{}}
\toprule
\plannerhead{Model} & \plannerhead{Cases} & \plannerhead{\shortstack{Accepted\\first}} &
\plannerhead{\shortstack{Accepted\\within 8}} & \plannerhead{\shortstack{Capped\\tries}} &
\plannerhead{Exec.} & \plannerhead{Correct} \\
\midrule
\texttt{qwen2.5:7b} (local) & 10 & 0 & 0 & 8.0 & 0 & 0 \\
\texttt{qwen2.5:14b} (local) & 10 & 0 & 0 & 8.0 & 0 & 0 \\
GPT-4o mini & 10 & 0 & 0 & 8.0 & 0 & 0 \\
GPT-4.1 mini & 10 & 0 & 10 & 4.1 & 10 & 9 \\
GPT-4o & 10 & 7 & 8 & 3.1 & 8 & 7 \\
GPT-4.1 & 10 & 10 & 10 & 1.0 & 10 & 6 \\
\bottomrule
\end{tabular*}
\endgroup
\end{table}

\section{Example: Medical Diagnosis Consensus}
\label{sec:examples}

In this section, we illustrate the language with a consensus workflow for medical diagnosis (Listing~\ref{lst:consensus}). Two independent LLMs analyze patient notes to determine whether a diagnosis (e.g., sepsis) applies. Each makes an independent assessment, then they exchange verdicts and reasoning,
iterating until consensus or a round limit is reached, at which point the final
verdict is returned to the user.
\texttt{LLM1} owns the loop and broadcasts the continuation decision to \texttt{LLM2} at each iteration. Deadlock-freedom (Corollary~\ref{cor:deadlock-free}) guarantees that every partial execution can be extended to a complete one.

\begin{lstlisting}[caption={Consensus workflow (lifeline and variable declarations omitted).},label={lst:consensus}]
workflow diagnosis_consensus(notes: str @ User, diagnosis: str @ User) -> str {
    // Distribute notes to both LLMs
    msg User(notes, diagnosis) -> LLM1(notes, diagnosis)
    msg User(notes, diagnosis) -> LLM2(notes, diagnosis)

    // Independent initial assessments
    act LLM1 : (verdict, reason) = assess(notes, diagnosis)
    act LLM2 : (verdict, reason) = assess(notes, diagnosis)
    msg LLM2(verdict) -> LLM1(other_verdict)
    act LLM1 : agreed = check_agreement(verdict, other_verdict)

    // Consensus loop: exchange, reconsider, check (LLM1 owns loop)
    while (not agreed and trials < max_rounds)@LLM1 {
        msg LLM1(verdict, reason) ->
            LLM2(other_verdict, other_reason)
        msg LLM2(verdict, reason) ->
            LLM1(other_verdict, other_reason)
        // Reconsider given other's reasoning
        act LLM1 : (verdict, reason) = reconsider(notes, diagnosis,
                                            verdict, reason,
                                            other_verdict, other_reason)
        act LLM2 : (verdict, reason) = reconsider(notes, diagnosis,
                                            verdict, reason,
                                            other_verdict, other_reason)

        // LLM2 sends updated verdict to LLM1 for agreement check
        msg LLM2(verdict) -> LLM1(other_verdict)

        // Check agreement and advance trial counter
        act LLM1 : agreed = check_agreement(verdict, other_verdict)
        act LLM1 : trials = inc_trials(trials)
    } exit { epsilon }

    // Final result computed locally, then sent unconditionally
    act LLM1 : result = choose_result(verdict, agreed)
    msg LLM1(result) -> User(result)
    return result @ User
}
\end{lstlisting}

The LLM action definition of \texttt{assess} is shown in Listing~\ref{lst:consensus-actions}.
The action \texttt{reconsider} is analogous.
The remaining actions \texttt{check\_agreement},
\texttt{inc\_trials}, and
\texttt{choose\_result} are simple (non-LLM) functions.
\(\mathtt{LLM1}\) owns the loop.
Both LLMs appear in the body, so \(\mathtt{LLM2}\) receives a control broadcast at each iteration to learn
whether execution continues or exits.
The initial assessments are structurally independent: both LLMs act on the
same notes without seeing each other's reasoning before the first verdict
exchange. This independence is visible in the global workflow and requires
no additional annotation or runtime enforcement.

\begin{lstlisting}[caption={Representative action definitions for the consensus workflow.},label={lst:consensus-actions}]
llm assess(notes: str, diag: str) -> (verdict: str, reason: str) {
    system: "You are a medical expert. Analyze the notes and
              determine if the diagnosis applies.
              Return verdict (yes/no/unknown) and your reasoning."
    user: "Notes: {{notes}}\nDiagnosis: {{diag}}"
    parse: json //{"verdict": "yes"/"no"/"unknown", "reason": "..."}
}
\end{lstlisting}

Additional executable examples, including a planning-based parallel arithmetic evaluator, are included with ZipperGen.

\paragraph{Verifiable properties.}
The workflow is deadlock-free by construction. Assuming \texttt{trials} starts at $0$
and \texttt{inc\_trials} increments it by one, at most \texttt{max\_rounds}
reconsideration rounds run, a property of the action implementations and their
dataflow, not of the abstract coordination semantics, and not a termination guarantee
for individual LLM calls. The final action \texttt{choose\_result} returns
the agreed verdict or the conservative fallback \texttt{unknown}.

\section{Conclusion}
\label{sec:conclusion}

We presented a domain-specific language for multi-agent LLM coordination
based on message sequence charts, together with a syntax-directed projection that
derives correct local programs from a global workflow specification.
The projection is always defined, requires no realizability check, and guarantees
deadlock-freedom by construction.
The framework further extends to runtime-generated workflows: a planning LLM
produces a workflow that is accepted by the parser and static validator, which is then projected and executed
with the same correctness guarantees as a statically written program.

Future work includes a formal treatment of stale or non-returning LLM calls,
where timeout mechanisms and their semantic consequences will be essential.
A second direction is the integration of coregions~\cite{itu-msc}, which relax
the total order on receptions within a lifeline and allow a more flexible
treatment of concurrency.
Finally, the precise operational semantics of the framework makes it a natural
setting for runtime verification: monitors can be derived from the global workflow
and used to enforce safety properties at the boundary of the non-verifiable LLM
components, providing trust guarantees that static analysis alone cannot offer.
Recent work on trace-based assurance for agentic systems~\cite{PaduraruBS26}
illustrates the kind of runtime contracts and failure classes that such
monitors could target.


\bibliographystyle{splncs04}
\bibliography{lit}


\clearpage
\appendix
\renewcommand{\theHsection}{appendix.\Alph{section}}
\section{Proofs}
\label{app:proofs}

\subsection{Proof of Lemma~\ref{lem:concat-msc} (Concatenation with a Complete Prefix)}
For the MSC part, consider \(M_1=(u_A)_A\) complete and \(M_2=(v_A)_A\) an MSC.
Every receive event of \(M_1\circ M_2\) lies either in the prefix \(M_1\) or
in the suffix \(M_2\). Prefix receives are matched inside \(M_1\) because
\(M_1\) is complete. Suffix receives are matched inside \(M_2\) because \(M_2\)
is an MSC. Hence FIFO matching is defined on \(M_1\circ M_2\), and label
compatibility is inherited componentwise from \(M_1\) and \(M_2\).

For acyclicity, note that all cross-component causality edges go from events of
\(M_1\) to events of \(M_2\): local successor edges cross the concatenation
boundary only in that direction, and message-matching edges cannot point from
\(M_2\) back into \(M_1\) because every send in \(M_1\) is already matched
inside the complete prefix \(M_1\). Thus a cycle in \(M_1\circ M_2\) would
restrict to a cycle in \(M_1\) or \(M_2\), impossible since both are MSCs.
Therefore \(M_1\circ M_2\) is an MSC.

If \(M_2\) is complete as well, then every send in the concatenation is matched
inside its own factor, so \(M_1\circ M_2\) is complete.
\qed

\subsection{Proof of Lemma~\ref{lem:strip-complete-prefix} (Stripping a Complete Prefix)}
For the MSC part, every receive event occurring in the suffix \(N\) is matched
in \(C\circ N\). On each channel, completeness of \(C\) means that \(C\) contains
the same number of sends and receives. Therefore every suffix receive has a
channel index strictly beyond those of the prefix receives and is matched to the
corresponding suffix send. Hence all suffix receives are matched by suffix sends,
so FIFO matching is defined on \(N\). Label compatibility is
inherited from \(C\circ N\). Acyclicity is preserved because the causality
relation of \(N\) is the restriction of the causality relation of \(C\circ N\)
to suffix events.

If \(C\circ N\) is complete, then every suffix send is matched in \(C\circ N\).
By the same channel-count argument, every suffix send has a channel index
strictly beyond those of the prefix sends and is matched to the corresponding
suffix receive. Therefore every suffix send is matched inside \(N\), so
\(N\) is complete.
\qed

\subsection{Proof of Lemma~\ref{lem:local-prefix-free}}
By structural induction on $S$.
The atomic cases are singletons and therefore prefix-free.

For sequential composition, let
\[
uv,\ u'v' \in \ltr{A}{S_1;S_2}
\]
with \(u,u'\in\ltr{A}{S_1}\) and \(v,v'\in\ltr{A}{S_2}\), and assume
\(uv\preceq u'v'\).
Then \(u\) and \(u'\) are both prefixes of the word \(u'v'\), hence they
are comparable by \(\preceq\). By IH on \(S_1\), the language
\(\ltr{A}{S_1}\) is prefix-free, so \(u=u'\). It follows that
\(v\preceq v'\), and IH on \(S_2\) gives \(v=v'\). Thus
\(uv=u'v'\).

For each control construct, we distinguish two projection cases: the owner case, where the
true and false branches begin with distinct choice letters (namely \(\IfT{c}{A}\)
versus \(\IfF{c}{A}\)), and the recipient case, where they begin with
distinct receive events (\(\mathtt{recv}\ A(\vec{y}[\top/y_1])\leftarrow B\)
versus \(\mathtt{recv}\ A(\vec{y}[\bot/y_1])\leftarrow B\)).
In both cases the argument is the same: two traces starting with different
initial letters are incomparable by \(\preceq\), and if they begin with the
same letter, prefix-freeness follows from IH on the selected branch.

For the two while-cases, every complete trace has the form
\[
d^\top u_1 \cdots d^\top u_k d^\bot v
\]
where \(d^\top\) and \(d^\bot\) are the corresponding top/bot decision letters,
\(u_1,\ldots,u_k\in\ltr{A}{S_{\mathit{body}}}\), and
\(v\in\ltr{A}{S_{\mathit{exit}}}\).
Compare the common iterations from left to right. If two corresponding body
traces differ, the induction hypothesis on \(S_{\mathit{body}}\) makes them
prefix-incomparable, so the complete traces are prefix-incomparable. Otherwise
all common body traces agree. If the iteration counts differ, the next letters
are then a \(\bot\)-decision and a \(\top\)-decision, so again neither word is a
prefix of the other. If the iteration counts agree, prefix-freeness follows from
the induction hypothesis on \(S_{\mathit{exit}}\) for the exit traces. Hence the
complete traces are prefix-free.
\qed

\subsection{Proof of Lemma~\ref{lem:erase-preserves} (Erasure Preserves MSC Structure)}
The map \(\mathsf{erase}\) deletes only control-broadcast send/receive events
introduced by projection and leaves all other local letters unchanged.
Therefore the relative order of the surviving events is preserved.

If \(M\) is an MSC, every receive event that survives erasure was already a
non-control receive in \(M\), and its matching send also survives erasure.
Conversely, deleting matched control send/receive pairs cannot create unmatched
receives among the remaining events. Label compatibility is inherited from
\(M\). For acyclicity, note that deleting an event makes its surviving neighbors
adjacent, so \(\mathsf{erase}(M)\) is not literally a subrelation of \(M\); rather,
every local-successor edge between surviving events corresponds to a nonempty
local-order path in \(M\), and matched control send/receive pairs are removed in
lockstep while the FIFO matching of the surviving user events is preserved. Hence
every causal cycle in \(\mathsf{erase}(M)\) would induce a causal cycle in \(M\),
which is impossible.
Therefore \(\mathsf{erase}(M)\) is an MSC.

If \(M\) is complete, then every surviving send event in \(\mathsf{erase}(M)\)
was already matched in \(M\) by a surviving receive event, since only control
broadcast pairs are deleted. Thus every send in \(\mathsf{erase}(M)\) is still
matched, so we have that \(\mathsf{erase}(M)\) is complete.
\qed

\subsection{Boundary-Aligned Sequential Factorization}

\begin{lemma}[Boundary-aligned sequential factorization \leanproof{https://github.com/zippergen-io/zippergen-lean/blob/isola-camera-ready/isola/MSCAgents/LocalSemantics.lean\#L307}]
\label{lem:boundary-aligned}
Let \(A\in\mathscr L\), let \(S_1,S_2 \in \mathsf{LocProg}_A\) be local programs at lifeline \(A\),
and let \(u\in\lpref{A}{S_1;S_2}\). Then there exist words
\(u_1,u_2\in\Sigma_A^*\) such that
\[
u=u_1u_2,\qquad
u_1\in\lpref{A}{S_1},\qquad
u_2\in\lpref{A}{S_2},
\]
and, moreover,
\[
u_2\neq\varepsilon ~\Longrightarrow~ u_1\in\ltr{A}{S_1}.
\]
\end{lemma}

\begin{proof}[Proof of Lemma~\ref{lem:boundary-aligned}]
Since \(u\in\lpref{A}{S_1;S_2}\), there exists \(w\in\ltr{A}{S_1;S_2}\) with
\(u\preceq w\). By the definition of the sequential composition semantics,
choose \(x\in\ltr{A}{S_1}\) and \(y\in\ltr{A}{S_2}\) such that \(w=xy\).
Since \(u\) is a prefix of the concatenation \(xy\), there are two cases.

If \(u\preceq x\), set \(u_1\df u\) and \(u_2\df\varepsilon\). Then
\(u_1\in\lpref{A}{S_1}\) and \(u_2\in\lpref{A}{S_2}\).

Otherwise, \(u\) crosses the \(S_1/S_2\)-boundary, so \(u=xu_2\) for some
prefix \(u_2\preceq y\). Set \(u_1\df x\). Then
\(u_1\in\ltr{A}{S_1}\) and \(u_2\in\lpref{A}{S_2}\). This yields the
required factorization.
\qed
\end{proof}

\subsection{Proof of Lemma~\ref{lem:uniform-zipper} (Zipper lemma)}
We proceed by structural induction on $P$.
Throughout the proof, $M=(u_Xv_X)_X$ is the MSC from the hypotheses.
Condition~3 of $\mathsf{ZipPost}$ (that $(u_Xv_X)_X$ is an MSC) is therefore
always satisfied by $M$ itself and requires no further argument.

\paragraph{Case $P=\varepsilon$.}
Then $u_X=\varepsilon$ for every $X$ (because
$u_X\in\lpref{X}{\proj{X}{\varepsilon}}=\{\varepsilon\}$).
Set $\bar u_X\df\varepsilon$. Hence
$\bar U=(\bar u_X)_X$ is a complete MSC and
$\bar M=(\bar u_Xv_X)_X=(v_X)_X=M$ is an MSC.
Thus \(\mathsf{ZipPost}_{\varepsilon}(U,\bar U,V)\) holds.

\paragraph{Cases $P=\mathtt{msg}\dots$, $P=\mathtt{act}\dots$.}
Assume $M=(u_Xv_X)_X$ satisfies the hypotheses.

For each $X$, define
\[
\bar u_X\df
\begin{cases}
u_X & \text{if } v_X\neq\varepsilon,\\[2mm]
u_Xt_X & \text{if } v_X=\varepsilon,
\end{cases}
\]
where $t_X$ is chosen such that
$u_Xt_X\in\ltr{X}{\proj{X}{P}}$.
Such a word exists because $u_X\in\lpref{X}{\proj{X}{P}}$.
Let $\bar U\df(\bar u_X)_X$ and
$\bar M\df(\bar u_Xv_X)_X$.

First, $\bar U$ is a complete MSC.
Indeed, for each $X$ we have
$\bar u_X\in\ltr{X}{\proj{X}{P}}$ by construction. Hence, for an action
statement, $\bar U$ is the corresponding canonical one-event MSC; for a
message statement, it is the canonical MSC consisting of the matching send
and receive of that statement. In both cases $\bar U$ is a complete MSC.

We show that $\bar M$ is an MSC.
First, $\lhd_{\bar M}$ is defined: pairs already matched in $M$
remain matched because FIFO matching is defined by channel indices and we
only append events at the end of local words.
Any new receives appear only in some $t_B$ (hence $v_B=\varepsilon$).
Fix such a receive on a channel $A\to B$.
Since $\bar u_A\in\ltr{A}{\proj{A}{P}}$ and
$\bar u_B\in\ltr{B}{\proj{B}{P}}$,
the projection of $P$ ensures that the corresponding send
appears in $\bar u_A$, hence the receive is matched.

Second, the causality relation of $\bar M$ (local order together with
FIFO constraints) is acyclic.
Here $P$ is a single action or message statement, so only its own events are
appended. If $P$ is an action statement, the only possible new event is an
isolated local action on a lifeline whose suffix is empty; it lies on no message
edge and cannot close a cycle. If $P=\mathtt{msg}\ A(\vec{x})\to B(\vec{y})$, the
only possible new events are its canonical send on $A$ and receive on $B$. If the
receive is newly appended, then $v_B=\varepsilon$, so this receive is the final
event on $B$ and has no outgoing causal edge. The case in which the receive is
already present but the send is newly appended is excluded because $M$ would then
contain an unmatched receive. Thus every new matched message edge ends in a
terminal receive and cannot lie on a causal cycle.
 Thus $\bar M$ is an MSC.
 
 Thus \(\mathsf{ZipPost}_{P}(U,\bar U,V)\) holds.

\paragraph{Auxiliary notation (decision prefixes).}
Let $P$ be a control construct (if or while) with decider $B$, and fix a branch value $\nu\in\{\top,\bot\}$.
Let $\mathcal R$ be the recipient set of that construct.
Write $\gamma_B^\nu$ for the \emph{choice letter} that $B$ emits when it decides~$\nu$:
this is $\IfT{c}{B}$ (resp.\ $\IfF{c}{B}$) for an if-construct,
and $\WhT{c}{B}$ (resp.\ $\WhF{c}{B}$) for a while-construct.
Define the \emph{decision-prefix word} $d_X^\nu$ on each lifeline $X$ by
\[
d_X^\nu \df
\begin{cases}
\gamma_B^\nu\cdot\prod_{A\in\mathcal R}^{\sqsubset} \mathtt{send}\ B(\nu,\kappa_{\mathsf{ctrl}}^P)\to A
& \text{if } X=B,\\[0.3em]
\mathtt{recv}\ X(\nu,\kappa_{\mathsf{ctrl}}^P)\leftarrow B
& \text{if } X\in \mathcal R,\\[0.3em]
\varepsilon
& \text{if } X\notin \mathcal R\cup\{B\}.
\end{cases}
\]
Note that the owner's component now starts with the choice letter $\gamma_B^\nu$,
which records the branch taken before any control message is sent.
In the if-case, this ensures the decomposition
$\proj{B}{P}=d_B^\nu;\proj{B}{Q_\nu}$ holds on the decider lifeline,
where $Q_\nu$ denotes the branch subprogram of $P$ corresponding to $\nu$.
Let $D^\nu \df (d_X^\nu)_{X\in\mathscr L}$.

\begin{lemma}[Decision-Prefix Stripping \leanproof{https://github.com/zippergen-io/zippergen-lean/blob/isola-camera-ready/isola/MSCAgents/BroadcastMSC.lean\#L321}]
\label{lem:strip-decision-prefix}
Let \(P\) be a projected if- or while-construct with decision-prefix tuple
\(D^\nu=(d_X^\nu)_X\) for some \(\nu\in\{\top,\bot\}\).
Let \(M=(w_X)_X\) be an MSC such that, for every lifeline \(X\),
\[
w_X\prec d_X^\nu
\qquad\text{or}\qquad
w_X=d_X^\nu s_X
\]
for some word \(s_X\).
Define
\[
r_X \df
\begin{cases}
\varepsilon & \text{if } w_X\prec d_X^\nu,\\
s_X & \text{if } w_X=d_X^\nu s_X,
\end{cases}
\qquad
M' \df (r_X)_X.
\]
Then \(M'\) is an MSC.
\end{lemma}

\begin{proof}[Proof of Lemma~\ref{lem:strip-decision-prefix}]
The deleted events on each lifeline form an initial segment of the
decision-prefix block \(d_X^\nu\).

We check that no unmatched receive is created.
Consider a control channel \(B\to A\) (owner to recipient).
The owner's send \(\mathtt{send}\ B(\nu,\kappa_{\mathsf{ctrl}}^P)\to A\)
appears in \(d_B^\nu\), and the recipient's matching receive
\(\mathtt{recv}\ A(\nu,\kappa_{\mathsf{ctrl}}^P)\leftarrow B\) is the sole
event of \(d_A^\nu\).
If the send is present in \(w_B\), it lies in \(d_B^\nu\) and is absent
from \(r_B\).
The receive is either also present in \(w_A\) (so \(w_A=d_A^\nu t_A\) and it
is likewise absent from \(r_A\)), or absent from \(w_A\) entirely.
In both sub-cases no unmatched receive appears in \(M'\).
If the send is absent (\(w_B\prec d_B^\nu\)), then \(M\) being an MSC ensures
the receive is absent from \(w_A\) as well (otherwise it would be unmatched
in \(M\)), so nothing changes for \(A\).
In all cases no unmatched receive is introduced.
Since \(d_X^\nu\) consists entirely of control events (the choice letter
and control sends/receives), non-control events appear only in the suffix
\(s_X\) and are fully preserved in \(M'\), hence no unmatched receive or
send is introduced on non-control channels either.
On each control channel, stripping removes either a matched initial
control send/receive pair, or a control send whose receive has not yet
occurred. In the latter case, the recipient has not passed its initial
control receive and therefore has no later events on that channel.
Consequently, the FIFO indices of surviving sends and receives are shifted
equally, and every surviving receive remains matched to the same surviving
send as in \(M\). Label compatibility is therefore preserved.

For acyclicity, every local-successor edge in \(M'\) corresponds to a
nonempty local-order path in \(M\), while every surviving message edge is
also a message edge of \(M\). Hence every causal cycle in \(M'\) would induce
a causal cycle in \(M\), contradicting that \(M\) is an MSC.
Hence \(M'\) is an MSC.
\qed
\end{proof}

\paragraph{If-construct case.}
Let
\[
P=\mathtt{if}\ c@B\ \mathtt{then}\ Q_{\top}\ \mathtt{else}\ Q_{\bot}.
\]
Inspect the decision-prefix events of this if-construct that occur in the
prefix tuple $U=(u_X)_X$. If any such event is present, the owner's choice
letter $\gamma_B^\nu$ and FIFO matching of the corresponding initial control
messages determine a unique branch value $\nu\in\{\top,\bot\}$, consistently
across all participating lifelines.
If no such event occurs in $U$, both values are compatible and we choose
one arbitrarily. In this situation, no participating lifeline can already have
entered the continuation of either branch: by the local trace semantics of the
projected if-construct, every such local trace starts with the corresponding
decision-prefix word \(d_X^\nu\).
Unfolding the projection and local prefix semantics with this value yields, for
each $X$:
\[
u_X\prec d_X^\nu
\qquad\text{or}\qquad
u_X=d_X^\nu u'_X\ \text{with }\ u'_X\in\lpref{X}{\proj{X}{Q_\nu}}.
\]
These are exactly the two prefix forms of the sequential decomposition
$\proj{X}{P}=d_X^\nu;\proj{X}{Q_\nu}$.

For each $X$, define the remainder
\[
r_X \df
\begin{cases}
\varepsilon & \text{if } u_X\prec d_X^\nu,\\
u'_X & \text{if } u_X=d_X^\nu u'_X\ \text{for some }u'_X\in\lpref{X}{\proj{X}{Q_\nu}}.
\end{cases}
\]
Moreover, for each $X$, if $v_X\neq\varepsilon$, then
$u_X\in\ltr{X}{\proj{X}{P}}$ by hypothesis.
Therefore the first alternative $u_X\prec d_X^\nu$ is impossible, hence
$u_X=d_X^\nu u'_X$ and $u'_X\in\ltr{X}{\proj{X}{Q_\nu}}$.
Apply Lemma~\ref{lem:strip-decision-prefix} to \(M=(u_Xv_X)_X\) with
\(w_X\df u_Xv_X\).
The two conditions of the lemma are satisfied: \(M\) is an MSC by hypothesis,
and each \(w_X\) either satisfies \(w_X\prec d_X^\nu\) (which forces
\(v_X=\varepsilon\) by the Zipper side condition, so \(w_X=u_X\prec d_X^\nu\))
or satisfies \(w_X=d_X^\nu(u'_Xv_X)\) (when \(u_X=d_X^\nu u'_X\)).
The lemma yields that
\[
M'\df(r_Xv_X)_X
\]
is an MSC.
Since $\proj{X}{P}=d_X^\nu;\proj{X}{Q_\nu}$, we have
$r_X\in\lpref{X}{\proj{X}{Q_\nu}}$ for every $X$;
and if $v_X\neq\varepsilon$ then $r_X=u'_X\in\ltr{X}{\proj{X}{Q_\nu}}$.
Apply IH to \(Q_\nu\) with respect to the factorization
$M'=(r_Xv_X)_X$.
We obtain a completion $\bar U_\nu\df(\bar r_X)_X$ such that
\[
\mathsf{ZipPost}_{Q_\nu}((r_X)_X,\bar U_\nu,V).
\]
Define
\[
\bar u_X\df d_X^\nu \bar r_X.
\]
We verify the four conditions of $\mathsf{ZipPost}_{P}(U,\bar U,V)$.
\emph{Condition~1.}
Since $r_X\preceq \bar r_X$ by Condition~1 of $\mathsf{ZipPost}_{Q_\nu}$, we have
$u_X\preceq \bar u_X$ (either $u_X\prec d_X^\nu\preceq\bar u_X$, or
$u_X=d_X^\nu u'_X=d_X^\nu r_X\preceq d_X^\nu\bar r_X=\bar u_X$).
Also $\bar u_X=d_X^\nu\bar r_X\in\ltr{X}{\proj{X}{P}}$ since
$\bar r_X\in\ltr{X}{\proj{X}{Q_\nu}}$.
If $v_X\neq\varepsilon$, then $\bar r_X=r_X=u'_X$ by Condition~1 of
$\mathsf{ZipPost}_{Q_\nu}$, so $\bar u_X=d_X^\nu u'_X=u_X$.
\emph{Condition~2.}
$\bar U_\nu=(\bar r_X)_X$ is a complete MSC by Condition~2 of
$\mathsf{ZipPost}_{Q_\nu}$, and $D^\nu$ is a complete MSC by construction,
so
\[
\bar U=D^\nu\circ \bar U_\nu
\]
is a complete MSC by Lemma~\ref{lem:concat-msc}.
\emph{Condition~3.}
$(u_Xv_X)_X=M$ is an MSC by hypothesis.
\emph{Condition~4.}
\[
(\bar u_Xv_X)_X
=
D^\nu \circ (\bar r_Xv_X)_X
\]
is an MSC by Lemma~\ref{lem:concat-msc}, since $D^\nu$ is a complete MSC and
$(\bar r_Xv_X)_X$ is an MSC by Condition~4 of $\mathsf{ZipPost}_{Q_\nu}$.
Thus \(\mathsf{ZipPost}_{P}(U,\bar U,V)\) holds.

\paragraph{While-construct case.}
Let
\[
P\df\mathtt{while}\ c@B\ \mathtt{do}\ Q\ \mathtt{exit}\ E.
\]
We argue by a secondary induction on the number \(m\ge 0\) of
already-started decision blocks of this while-construct visible in the prefix
tuple \(U=(u_X)_X\), equivalently in the unique decomposition of the owner's
prefix \(u_B\).
The control sends and receives of these blocks carry the tag
\(\kappa_{\mathsf{ctrl}}^P\), distinct from the tags of all constructs nested
inside \(Q\) or \(E\). The owner's choice letters carry no tag, and a while-construct
nested inside \(Q\) may share the same condition and owner. The count \(m\) is
nevertheless well defined, because the decomposition of the owner's prefix
\(u_B\in\lpref{B}{\proj{B}{P}}\) into decision blocks, completed body traces,
and a partial remainder is unique: a completed body trace can never be a
proper prefix of another word of \(\ltr{B}{\proj{B}{Q}}\) by prefix-freeness
(Lemma~\ref{lem:local-prefix-free}), so the boundaries between iterations are
determined. Thus \(m\) is the number of unfoldings of this while-construct
visible in that decomposition. The associated control sends and receives may
still be only partially present.
(In the Lean formalization, the secondary induction is instead performed on
the length of the owner's prefix \(u_B\), which sidesteps this counting
argument.)
The owner's decision sequence is determined by this decomposition; by FIFO
matching, any visible control sends and receives are consistent with it.
The resulting sequence of decisions has the form
\[
\nu_1\ldots \nu_m \in \{\top\}^* \cup \{\top\}^*\bot .
\]

\smallskip
\emph{Base case \(m=0\).}
Then no decision-prefix event (choice letter, control send, or control receive)
of this while-construct has yet started in the tuple.
In particular, for every lifeline \(X\) that participates in this
while-construct, one has \(u_X=\varepsilon\).
Equivalently, no participating lifeline has yet entered either the loop body or
the exit continuation, because every participating local trace of the projected
while-program begins with the first decision-prefix block.
Moreover, if \(v_X\neq\varepsilon\), then by the zipper side condition
\(u_X\in\ltr{X}{\proj{X}{P}}\) holds.
Thus, such an \(X\) cannot participate in the
while-construct, because every complete local trace of a participating
projection of \(P\) begins with a decision-prefix event (a choice letter on
the owner, a control receive on each recipient). Since here \(u_X=\varepsilon\),
the only way \(u_X\in\ltr{X}{\proj{X}{P}}\) can hold is that
\(\proj{X}{P}=\varepsilon\), i.e., \(X\) does not participate in the
while-construct at all. In particular, \(X\notin\LS{E}\), so
\(\proj{X}{E}=\varepsilon\).

We choose the first decision value to be \(\bot\).
Let
\[
r_X\df\varepsilon
\qquad\text{for all }X\in\mathscr L,
\]
and consider the tuple
\[
M'\df(r_Xv_X)_X=(v_X)_X.
\]
For every \(X\), one has
\[
r_X=\varepsilon\in\lpref{X}{\proj{X}{E}},
\]
and if \(v_X\neq\varepsilon\), then \(X\) does not participate in the
while-construct, hence also
\[
r_X=\varepsilon\in\ltr{X}{\proj{X}{E}}.
\]
Therefore the structural induction hypothesis applies to \(E\) and
\(M'=(r_Xv_X)_X\). We obtain a tuple
\[
\bar R=(\bar r_X)_X
\]
such that
\[
\mathsf{ZipPost}_{E}((r_X)_X,\bar R,V).
\]
Now define
\[
\bar u_X\df d_X^\bot\,\bar r_X
\qquad\text{for all }X\in\mathscr L,
\]
and let \(\bar U\df(\bar u_X)_X\).
Since \(u_X=\varepsilon\preceq \bar u_X\), Condition~1a holds.
If \(v_X\neq\varepsilon\), then \(X\) does not participate in the
while-construct (shown above), so \(d_X^\bot=\varepsilon\) and
\(\bar u_X=\varepsilon=u_X\), and Condition~1c holds.
\(\bar u_X\in\ltr{X}{\proj{X}{P}}\) for every \(X\) (Condition~1b), because a
\(\bot\)-decision followed by an \(E\)-trace is exactly a complete local trace
of the projected while-program.
Furthermore,
\[
\bar U=D^\bot\circ \bar R
\]
is a complete MSC (Condition~2), and
\[
\bar M=(\bar u_Xv_X)_X
=
D^\bot\circ (\bar r_Xv_X)_X
\]
is an MSC by Lemma~\ref{lem:concat-msc}, since \(D^\bot\) is a complete MSC
(Condition~4).
Hence \(\mathsf{ZipPost}_{P}(U,\bar U,V)\) holds.
This yields the required completion in the base case.

\smallskip
\emph{Induction step \(m>0\).}
Let \(\nu\df\nu_1\in\{\top,\bot\}\) be the first visible decision value.
By unfolding the definitions of projection for \(P\) and of
\(\lpref{X}{\cdot}\), each local prefix \(u_X\in\lpref{X}{\proj{X}{P}}\) is
of one of the following forms:
\begin{itemize}
\item $u_X\prec d_X^\nu$,
\item
$u_X=d_X^\top u'_X$ with $u'_X\in\lpref{X}{\proj{X}{Q;P}}$,
\item
$u_X=d_X^\bot u'_X$ with $u'_X\in\lpref{X}{\proj{X}{E}}$.
\end{itemize}
Since \(\nu\) is the first visible decision value forced by the decision-prefix
events already present in \(M\), only the branch corresponding to this \(\nu\)
can occur.

Define the remainder
\[
r_X \df
\begin{cases}
\varepsilon & \text{if } u_X\prec d_X^\nu,\\
u'_X & \text{if } u_X=d_X^\nu u'_X.
\end{cases}
\]
Moreover, if \(v_X\neq\varepsilon\), then by hypothesis
\(u_X\in\ltr{X}{\proj{X}{P}}\), so the alternative \(u_X\prec d_X^\nu\) is
impossible. Thus whenever \(v_X\neq\varepsilon\), the first decision block has
already been fully consumed on \(X\).
Apply Lemma~\ref{lem:strip-decision-prefix} to \(M=(u_Xv_X)_X\) with
\(w_X\df u_Xv_X\): since \(v_X\neq\varepsilon\) implies \(u_X=d_X^\nu u'_X\),
each \(w_X\) satisfies the prefix condition.
The lemma yields that
\[
M'\df(r_Xv_X)_X
\]
is again an MSC.

We now distinguish the two possible values of \(\nu\).

\smallskip
\emph{Subcase \(\nu=\bot\).}
Then for every \(X\),
\[
r_X\in\lpref{X}{\proj{X}{E}},
\]
and whenever \(v_X\neq\varepsilon\), since \(u_X\in\ltr{X}{\proj{X}{P}}\)
and \(u_X=d_X^\bot u'_X\), every complete local trace of \(\proj{X}{P}\)
starting with \(d_X^\bot\) continues in \(\ltr{X}{\proj{X}{E}}\), hence
\[
r_X=u'_X\in\ltr{X}{\proj{X}{E}}.
\]
Hence the structural induction hypothesis applies to \(E\) and
\(M'=(r_Xv_X)_X\). We obtain a completion
\[
\bar R=(\bar r_X)_X
\]
such that
\[
\mathsf{ZipPost}_{E}((r_X)_X,\bar R,V).
\]
Define
\[
\bar u_X\df d_X^\bot\,\bar r_X,
\]
and let \(\bar U\df(\bar u_X)_X\).
For every \(X\), \(u_X\preceq \bar u_X\): if \(u_X=d_X^\bot r_X\) then
\(u_X=d_X^\bot r_X\preceq d_X^\bot \bar r_X=\bar u_X\) by Condition~1 of
\(\mathsf{ZipPost}_{E}\): if \(u_X\prec d_X^\bot\) then \(u_X\preceq d_X^\bot\preceq \bar u_X\).
If \(v_X\neq\varepsilon\), then \(u_X=d_X^\bot r_X\) (shown above) and
\(\bar r_X=r_X\) by Condition~1 of \(\mathsf{ZipPost}_{E}\), hence
\(\bar u_X=d_X^\bot \bar r_X=d_X^\bot r_X=u_X\).
Therefore Condition~1 of \(\mathsf{ZipPost}_{P}(U,\bar U,V)\) holds.
\(\bar u_X\in\ltr{X}{\proj{X}{P}}\) for every \(X\) (Condition~1b), and
\[
\bar U=D^\bot\circ \bar R
\]
is a complete MSC (Condition~2), while
\[
\bar M=(\bar u_Xv_X)_X
=
D^\bot\circ (\bar r_Xv_X)_X
\]
is an MSC by Lemma~\ref{lem:concat-msc}, since \(D^\bot\) is a complete MSC
(Condition~4).
Hence \(\mathsf{ZipPost}_{P}(U,\bar U,V)\) holds.

\smallskip
\emph{Subcase \(\nu=\top\).}
Then for every \(X\),
\[
r_X\in\lpref{X}{\proj{X}{Q;P}},
\]
and whenever \(v_X\neq\varepsilon\), since \(u_X\in\ltr{X}{\proj{X}{P}}\)
and \(u_X=d_X^\top u'_X\), every complete local trace of \(\proj{X}{P}\)
starting with \(d_X^\top\) continues in \(\ltr{X}{\proj{X}{Q;P}}\), hence
\[
r_X=u'_X\in\ltr{X}{\proj{X}{Q;P}}.
\]
By the prefix semantics of sequential composition, factor
\[
r_X=b_Xs_X
\]
with
\[
b_X\in\lpref{X}{\proj{X}{Q}}
\qquad\text{and}\qquad
s_X\in\lpref{X}{\proj{X}{P}}.
\]
By Lemma~\ref{lem:boundary-aligned}, we may choose this factorization
boundary-aligned, i.e.,
\[
s_X\neq\varepsilon ~\Longrightarrow~ b_X\in\ltr{X}{\proj{X}{Q}}.
\]
The remaining case needed for the IH side condition is \(s_X=\varepsilon\)
and \(v_X\neq\varepsilon\): then \(r_X=b_X\in\lpref{X}{\proj{X}{Q}}\), and
the condition established above (\(v_X\neq\varepsilon\Longrightarrow r_X\in\ltr{X}{\proj{X}{Q;P}}\))
gives \(r_X\in\ltr{X}{\proj{X}{Q;P}}\).
Since \(b_X=r_X\in\ltr{X}{\proj{X}{Q;P}}\), it factors as \(b_X=c_Xd_X\)
with \(c_X\in\ltr{X}{\proj{X}{Q}}\) and \(d_X\in\ltr{X}{\proj{X}{P}}\).
Let \(e_X\in\ltr{X}{\proj{X}{Q}}\) be a complete trace with
\(b_X\preceq e_X\) (which exists because \(b_X\in\lpref{X}{\proj{X}{Q}}\)).
Then \(c_X\preceq b_X\preceq e_X\) with \(c_X,e_X\in\ltr{X}{\proj{X}{Q}}\);
by prefix-freeness of \(\ltr{X}{\proj{X}{Q}}\), this forces \(c_X=e_X\),
hence \(d_X=\varepsilon\) and \(b_X=c_X\in\ltr{X}{\proj{X}{Q}}\).
Hence in all cases,
\[
s_Xv_X\neq\varepsilon ~\Longrightarrow~ b_X\in\ltr{X}{\proj{X}{Q}}.
\]

Now consider
\[
N\df(b_X(s_Xv_X))_X.
\]
Since \(N=(r_Xv_X)_X=M'\) is an MSC and
\[
b_X\in\lpref{X}{\proj{X}{Q}},
\]
with the required side condition established above, the
structural induction hypothesis applies to \(Q\) and \(N\).
Hence there exists a completion
\[
\bar U_Q=(\bar b_X)_X
\]
such that
\[
\mathsf{ZipPost}_{Q}((b_X)_X,\bar U_Q,(s_Xv_X)_X).
\]
In particular, \(M''\df(s_Xv_X)_X\) is an MSC.
Furthermore, if \(v_X\neq\varepsilon\), then \(s_Xv_X\neq\varepsilon\), hence
\(\bar b_X=b_X\in\ltr{X}{\proj{X}{Q}}\). Since
\[
b_Xs_X=r_X\in\ltr{X}{\proj{X}{Q;P}},
\]
factor \(b_Xs_X=c_Xd_X\) with \(c_X\in\ltr{X}{\proj{X}{Q}}\) and
\(d_X\in\ltr{X}{\proj{X}{P}}\).
Since \(b_X,c_X\in\ltr{X}{\proj{X}{Q}}\) are comparable by the prefix order
(as \(b_X\preceq c_Xd_X\)), prefix-freeness of \(\ltr{X}{\proj{X}{Q}}\) forces
\(c_X=b_X\), hence
\[
s_X=d_X\in\ltr{X}{\proj{X}{P}}.
\]
Thus \(v_X\neq\varepsilon\Longrightarrow s_X\in\ltr{X}{\proj{X}{P}}\),
i.e., the zipper side condition holds for \(M''=(s_Xv_X)_X\) with respect to \(P\).

It remains to justify the decrease of the secondary induction measure. Recall
that \(m\) counts the started owner-side decision blocks of this
while-construct in the owner's prefix \(u_B\), where \(\nu_1=\top\) in the
present subcase. The
first transformation,
\[
M \;\longrightarrow\; M'=(r_Xv_X)_X,
\]
removes the first owner choice \(\gamma_B^\top\), together with any visible
prefix of its decision block, by Lemma~\ref{lem:strip-decision-prefix}. The second
transformation,
\[
M' \;\longrightarrow\; M''=(s_Xv_X)_X,
\]
removes the \(Q\)-prefix \(b_X\) from each lifeline via the factorization
\(r_X=b_Xs_X\), where \(b_X\in\lpref{X}{\proj{X}{Q}}\) and
\(s_X\in\lpref{X}{\proj{X}{P}}\). Here \(Q\) is the loop body executed after
the first \(\top\)-decision, whereas \(P\) is the recursive continuation of
the same while-program after that body iteration.

Removing the first decision block and the following \(Q\)-prefix \(b_B\)
from \(u_B\) leaves the recursive prefix \(s_B\), whose unique outer-loop
decomposition contains exactly \(m-1\) started decision blocks; the suffix
\(v_B\) plays no role in this count.

The secondary induction hypothesis therefore applies to the instance with
prefix tuple \((s_X)_X\) and suffix tuple \(V=(v_X)_X\). Hence there
exists a completion
\[
\bar S=(\bar s_X)_X
\]
such that
\[
\mathsf{ZipPost}_{P}((s_X)_X,\bar S,V).
\]
Define
\[
\bar u_X\df d_X^\top\,\bar b_X\,\bar s_X,
\]
and let \(\bar U\df(\bar u_X)_X\).
For every \(X\), if \(u_X\prec d_X^\top\), then
\(u_X\preceq d_X^\top\preceq d_X^\top\bar b_X\bar s_X=\bar u_X\).
Otherwise \(u_X=d_X^\top r_X=d_X^\top b_Xs_X\), and Condition~1 of the two
postconditions \(\mathsf{ZipPost}_{Q}((b_X)_X,\bar U_Q,(s_Xv_X)_X)\) and
\(\mathsf{ZipPost}_{P}((s_X)_X,\bar S,V)\) gives \(b_X\preceq\bar b_X\) and
\(s_X\preceq\bar s_X\), hence again
\(u_X\preceq d_X^\top\bar b_X\bar s_X=\bar u_X\).
If \(v_X\neq\varepsilon\), the first case is impossible, and the two
postconditions give \(\bar b_X=b_X\) and \(\bar s_X=s_X\), so \(\bar u_X=u_X\).
Therefore Condition~1 of
\(\mathsf{ZipPost}_{P}(U,\bar U,V)\) holds.
Moreover, \(\bar u_X\in\ltr{X}{\proj{X}{P}}\) for every \(X\), because a
\(\top\)-decision followed by a \(Q\)-trace and then a \(P\)-trace is exactly a
complete local trace of the projected while-program.
Since \(\bar U_Q\) is a complete MSC and
\((\bar s_Xv_X)_X\) is an MSC (from \(\mathsf{ZipPost}_{P}((s_X)_X,\bar S,V)\)),
Lemma~\ref{lem:concat-msc} yields that
\((\bar b_X\bar s_Xv_X)_X=\bar U_Q\circ(\bar s_Xv_X)_X\) is an MSC.
Moreover,
\[
\bar U=D^\top\circ \bar U_Q\circ \bar S
\]
is a complete MSC, and
\[
\bar M=(\bar u_Xv_X)_X
=
D^\top\circ (\bar b_X\bar s_Xv_X)_X
\]
is an MSC by another application of Lemma~\ref{lem:concat-msc},
since \(D^\top\) is a complete MSC.

In both subcases, we have constructed a tuple \(\bar U=(\bar u_X)_X\) such
that the condition \(\mathsf{ZipPost}_{P}(U,\bar U,V)\) holds.
This concludes the while case.

\medskip
\paragraph{Composition case.}
Let $P=Q_1; Q_2$.
Since we have \(u_X\in\lpref{X}{\proj{X}{Q_1;Q_2}}\), apply
Lemma~\ref{lem:boundary-aligned} to obtain a factorization
\[
u_X=u_X^{(1)}u_X^{(2)},
\qquad
u_X^{(1)}\in\lpref{X}{\proj{X}{Q_1}},
\qquad
u_X^{(2)}\in\lpref{X}{\proj{X}{Q_2}},
\]
with the additional guarantee that
\[
u_X^{(2)}\neq\varepsilon ~\Longrightarrow~ u_X^{(1)}\in\ltr{X}{\proj{X}{Q_1}}.
\]
We now show that the two IH side conditions hold for this factorization.

\emph{First side condition} (\(u_X^{(2)}v_X\neq\varepsilon \Longrightarrow
u_X^{(1)}\in\ltr{X}{\proj{X}{Q_1}}\)).
If \(u_X^{(2)}\neq\varepsilon\), this follows directly from the guarantee of
Lemma~\ref{lem:boundary-aligned}.
If \(u_X^{(2)}=\varepsilon\) and \(v_X\neq\varepsilon\), then
\(u_X=u_X^{(1)}\) and the zipper hypothesis gives
\(u_X^{(1)}\in\ltr{X}{\proj{X}{Q_1;Q_2}}\).
Hence there exist \(w_1\in\ltr{X}{\proj{X}{Q_1}}\) and
\(w_2\in\ltr{X}{\proj{X}{Q_2}}\) with \(u_X^{(1)}=w_1w_2\).
Since \(u_X^{(1)}\in\lpref{X}{\proj{X}{Q_1}}\), there exists
\(t\in\ltr{X}{\proj{X}{Q_1}}\) with \(u_X^{(1)}\preceq t\).
Then \(w_1\preceq u_X^{(1)}\preceq t\) with \(w_1,t\in\ltr{X}{\proj{X}{Q_1}}\).
By prefix-freeness of \(\ltr{X}{\proj{X}{Q_1}}\), \(w_1=t\), so
\(u_X^{(1)}\preceq w_1\).
Combined with \(w_1\preceq u_X^{(1)}\), we get \(u_X^{(1)}=w_1\in\ltr{X}{\proj{X}{Q_1}}\)
and \(w_2=\varepsilon\).

\emph{Second side condition} (\(v_X\neq\varepsilon \Longrightarrow
u_X^{(2)}\in\ltr{X}{\proj{X}{Q_2}}\)).
Assume \(v_X\neq\varepsilon\).
Since \(v_X\neq\varepsilon\) implies \(u_X^{(2)}v_X\neq\varepsilon\), the
first side condition gives \(u_X^{(1)}\in\ltr{X}{\proj{X}{Q_1}}\).
The zipper hypothesis gives
\(u_X=u_X^{(1)}u_X^{(2)}\in\ltr{X}{\proj{X}{Q_1;Q_2}}
=\ltr{X}{\proj{X}{Q_1}}\cdot\ltr{X}{\proj{X}{Q_2}}\),
so there exist \(w_1'\in\ltr{X}{\proj{X}{Q_1}}\) and
\(w_2'\in\ltr{X}{\proj{X}{Q_2}}\) with \(u_X^{(1)}u_X^{(2)}=w_1'w_2'\).
From \(u_X^{(1)}u_X^{(2)}=w_1'w_2'\), one of \(w_1'\preceq u_X^{(1)}\) or
\(u_X^{(1)}\preceq w_1'\) holds (possibly both, when \(w_1'=u_X^{(1)}\)).
Since both \(u_X^{(1)}\) and \(w_1'\) lie in \(\ltr{X}{\proj{X}{Q_1}}\),
prefix-freeness forces \(w_1'=u_X^{(1)}\), hence
\(u_X^{(2)}=w_2'\in\ltr{X}{\proj{X}{Q_2}}\).

\smallskip

Write $U^{(1)}=(u_X^{(1)})_X$ and $W=(u_X^{(2)}v_X)_X$.
Apply IH to \(Q_1\) and
$M=(u_X^{(1)}(u_X^{(2)}v_X))_X$.
This yields \(\mathsf{ZipPost}_{Q_1}(U^{(1)},\bar U^{(1)},W)\).
In particular,
$W$ is an MSC by conditions~2 and~4 of \(\mathsf{ZipPost}_{Q_1}\) and
Lemma~\ref{lem:strip-complete-prefix}.
For the second IH application, assume $v_X\neq\varepsilon$. Then
$u_X^{(2)}\in\ltr{X}{\proj{X}{Q_2}}$ by the second side condition established above.
So the side condition for IH on \(Q_2\) holds.
Apply IH to \(Q_2\) and $W$.
This yields \(\mathsf{ZipPost}_{Q_2}(U^{(2)},\bar U^{(2)},V)\).
Setting
\[
\bar u_X\df \bar u_X^{(1)}\bar u_X^{(2)}
\]
gives the required completion for $P$.
Let \(\bar U\df(\bar u_X)_X\).
Completeness of $\bar U$ follows from completeness of the factors and
Lemma~\ref{lem:concat-msc}.
Moreover, since \(\bar U^{(1)}\) is a complete MSC and
\((\bar u_X^{(2)}v_X)_X\) is an MSC, Lemma~\ref{lem:concat-msc} yields that
\((\bar u_Xv_X)_X\) is an MSC. Therefore for every \(X\), \(u_X\preceq\bar u_X\):
the second IH gives \(u_X^{(2)}\preceq\bar u_X^{(2)}\) unconditionally;
when \(u_X^{(2)}v_X\neq\varepsilon\), the first IH additionally gives
\(\bar u_X^{(1)}=u_X^{(1)}\), so
\(u_X=u_X^{(1)}u_X^{(2)}\preceq u_X^{(1)}\bar u_X^{(2)}=\bar u_X^{(1)}\bar u_X^{(2)}=\bar u_X\);
when \(u_X^{(2)}=\varepsilon\) and \(v_X=\varepsilon\), the first IH gives
\(u_X=u_X^{(1)}\preceq\bar u_X^{(1)}\preceq\bar u_X\).
Moreover, if \(v_X\neq\varepsilon\), then
\(\bar u_X^{(2)}=u_X^{(2)}\) (second IH) and \(u_X^{(2)}v_X\neq\varepsilon\),
so \(\bar u_X^{(1)}=u_X^{(1)}\) (first IH), hence \(\bar u_X=u_X\). Thus
Condition~1 of \(\mathsf{ZipPost}_{Q_1;Q_2}(U,\bar U,V)\) holds, and so does
\(\mathsf{ZipPost}_{Q_1;Q_2}(U,\bar U,V)\).
\qed

\subsection{Auxiliary Consequences of \textsf{ZipPost}}

\begin{lemma}[\leanproof{https://github.com/zippergen-io/zippergen-lean/blob/isola-camera-ready/isola/MSCAgents/ZipperPost.lean\#L262}]
\label{lem:U-is-MSC}
If \(\mathsf{ZipPost}_{P}(U,\bar U,V)\) holds, then
\(U\) is an MSC.
\end{lemma}

\begin{proof}[Proof of Lemma~\ref{lem:U-is-MSC}]
Acyclicity and label compatibility are inherited from
\((u_Xv_X)_X\) (an MSC by Condition~3 of \(\mathsf{ZipPost}_P(U,\bar U,V)\))
by restriction to prefix events.

For the receive-matching condition, suppose that on some channel \(A\to B\),
the prefix \(u_B\) contains more receives than \(u_A\) contains sends. Since
\(\bar U\) is a complete MSC (Condition~2), the completion \(\bar u_A\) must
contain an additional send on that channel, hence \(\bar u_A\neq u_A\). By
Condition~1, this implies \(v_A=\varepsilon\), so the full word for \(A\) in
\((u_Xv_X)_X\) is just \(u_A\). But then \((u_Xv_X)_X\) has an unmatched
receive on channel \(A\to B\), contradicting Condition~3. Hence \(U\) has no
unmatched receives and is therefore an MSC.
\qed
\end{proof}

\begin{lemma}[\leanproof{https://github.com/zippergen-io/zippergen-lean/blob/isola-camera-ready/isola/MSCAgents/Correctness.lean\#L2514}]
\label{lem:complete-locals-complete-msc}
Let \(P\) be a well-typed global workflow, and let
\(M=(u_X)_{X\in\mathscr L}\) be a tuple such that
\[
u_X\in\ltr{X}{\proj{X}{P}}
\qquad\text{for every }X\in\mathscr L.
\]
If \(M\) is an MSC, then \(M\) is a complete MSC.
\end{lemma}

\begin{proof}[Proof of Lemma~\ref{lem:complete-locals-complete-msc}]
Apply Lemma~\ref{lem:uniform-zipper} with \(U=(u_X)_X\) and
\(V=(\varepsilon)_X\). Since each \(u_X\in\ltr{X}{\proj{X}{P}}\), one has in
particular \(u_X\in\lpref{X}{\proj{X}{P}}\), and because every \(v_X\) is
\(\varepsilon\), the zipper side condition is vacuous. As \(M=(u_X)_X\) is an
MSC by hypothesis, Lemma~\ref{lem:uniform-zipper} yields a tuple
\(\bar U=(\bar u_X)_X\) such that \(\mathsf{ZipPost}_{P}(U,\bar U,V)\) holds.

For every lifeline \(X\), Condition~1 of \(\mathsf{ZipPost}_P(U,\bar U,V)\) gives
\[
u_X\preceq \bar u_X
\qquad\text{and}\qquad
\bar u_X\in\ltr{X}{\proj{X}{P}}.
\]
Since complete local trace languages are prefix-free by
Lemma~\ref{lem:local-prefix-free}, it follows that \(u_X=\bar u_X\) for every
\(X\). Hence \(U=\bar U\). Now Condition~2 of
\(\mathsf{ZipPost}_P(U,\bar U,V)\) states that \(\bar U\) is a complete MSC.
Therefore we have that \(M=U\) is a complete MSC as well.
\qed
\end{proof}

\begin{corollary}[\leanproof{https://github.com/zippergen-io/zippergen-lean/blob/isola-camera-ready/isola/MSCAgents/ZipperPost.lean\#L342}]
\label{cor:U-complete}
If \(\mathsf{ZipPost}_{P}(U,\bar U,V)\) holds and
\(u_X\in\ltr{X}{\proj{X}{P}}\) for every lifeline \(X\), then \(U\) is a
complete MSC.
\end{corollary}

\begin{proof}[Proof of Corollary~\ref{cor:U-complete}]
By Lemma~\ref{lem:U-is-MSC}, the tuple \(U=(u_X)_X\) is an MSC. Since
\(u_X\in\ltr{X}{\proj{X}{P}}\) for every lifeline \(X\), Lemma~\ref{lem:complete-locals-complete-msc}
implies that \(U\) is a complete MSC.
\qed
\end{proof}

\subsection{Control Decomposition}

We use the decision-prefix notation \(D^\nu=(d_X^\nu)_X\) introduced in the
proof of Lemma~\ref{lem:uniform-zipper}.

\begin{lemma}[Control Decomposition \leanproof{https://github.com/zippergen-io/zippergen-lean/blob/isola-camera-ready/isola/MSCAgents/Correctness.lean\#L443}]
\label{lem:control-decompose}
Let \(P\) be a well-typed global workflow.

\begin{enumerate}
\item If
\[
P=\mathtt{if}\ c@B\ \mathtt{then}\ Q_\top\ \mathtt{else}\ Q_\bot
\]
and \(\hat M\in\sem{\mathcal D_P}\), then there exist a unique
\(\nu\in\{\top,\bot\}\) and an MSC \(\hat M_\nu\in\sem{\mathcal D_{Q_\nu}}\)
such that
\[
\hat M = D^\nu \circ \hat M_\nu.
\]

\item If
\[
P=\mathtt{while}\ c@B\ \mathtt{do}\ Q\ \mathtt{exit}\ E
\]
and \(\hat M\in\sem{\mathcal D_P}\), then there exist \(k\ge 0\),
\(\hat M_1,\ldots,\hat M_k\in\sem{\mathcal D_Q}\), and
\(\hat M_{\mathit{ex}}\in\sem{\mathcal D_E}\) such that
\[
\hat M =
\bigl(D^\top\circ \hat M_1\bigr)\circ\cdots\circ
\bigl(D^\top\circ \hat M_k\bigr)\circ
\bigl(D^\bot\circ \hat M_{\mathit{ex}}\bigr).
\]
\end{enumerate}
\end{lemma}

\begin{proof}[Proof of Lemma~\ref{lem:control-decompose}]
\emph{If-case.}
By Table~\ref{tab:projection-if}, for every lifeline \(X\) the complete local
trace language of \(\proj{X}{P}\) consists of words of the form
\(d_X^\nu\cdot t_X\) with \(\nu\in\{\top,\bot\}\) and
\(t_X\in\ltr{X}{\proj{X}{Q_\nu}}\).
Hence the \(X\)-component of \(\hat M\) starts with \(d_X^{\nu_X}\)
for some \(\nu_X\).
The value \(\nu_X\) is the same for every \(X\).
By Table~\ref{tab:projection-if}, the owner \(B\)'s complete local trace
begins with its choice event \(\gamma_B^\nu\) followed by the ordered control
broadcast of some \(\nu\in\{\top,\bot\}\) to all recipients.
By label compatibility of matched send/receive events (Definition~\ref{def:msc})
and FIFO ordering on each \(B\to X\) control channel, every recipient \(X\)'s
component of \(\hat M\) also begins with \(\nu\); idle agents have \(d_X^\nu=\varepsilon\).
Call this common value \(\nu\).

Write the \(X\)-component of \(\hat M\) as \(d_X^\nu\cdot t_X\)
(with \(t_X=\varepsilon\) for idle agents where \(d_X^\nu=\varepsilon\)) and
define \(\hat M_\nu\df(t_X)_X\), so \(\hat M=D^\nu\circ\hat M_\nu\).
Since \(D^\nu\) is a complete MSC by construction and \(\hat M\) is a complete
MSC, Lemma~\ref{lem:strip-complete-prefix} yields that \(\hat M_\nu\) is a
complete MSC.
Each component \(t_X\in\ltr{X}{\proj{X}{Q_\nu}}\), so
\(\hat M_\nu\in\sem{\mathcal D_{Q_\nu}}\).
Uniqueness of \(\hat M_\nu\) follows because concatenation with the fixed prefix \(D^\nu\) is injective: if \(D^\nu\circ N = D^\nu\circ N'\) then \(N=N'\).

\smallskip
\emph{While-case.}
We argue by induction on the number \(k\ge 0\) of \(\top\)-decision blocks
in the owner \(B\)'s component of \(\hat M\).
Here \(k\) is the number of outer-loop \(\top\)-iterations in the unique
decomposition of the owner's complete local trace into decision blocks,
complete \(Q\)-traces, and the final \(\bot\)-decision block followed by an
\(E\)-trace; uniqueness follows from prefix-freeness
(Lemma~\ref{lem:local-prefix-free}).

\emph{Base case \(k=0\).}
Since \(k=0\), the owner \(B\)'s component begins with \(d_B^\bot\) (a \(\bot\)-decision block).
By label compatibility (Definition~\ref{def:msc}) and FIFO ordering on each \(B\to X\) control
channel, every recipient \(X\)'s component of \(\hat M\) also begins with \(d_X^\bot\);
idle agents have \(d_X^\bot=\varepsilon\).
The tuple \(D^\bot=(d_X^\bot)_X\) is a complete MSC because the control sends on \(B\)
are matched to the corresponding control receives on each recipient by the FIFO matching of \(\hat M\).
By Table~\ref{tab:projection-while}, every complete local trace of
\(\proj{X}{P}\) that starts with a \(\bot\)-decision has the form
\(d_X^\bot\cdot e_X\) with \(e_X\in\ltr{X}{\proj{X}{E}}\).
Write the \(X\)-component of \(\hat M\) as \(d_X^\bot\cdot e_X\)
(with \(e_X=\varepsilon\) for idle agents where \(d_X^\bot=\varepsilon\)) and
define \(\hat M_{\mathit{ex}}\df(e_X)_X\), so \(\hat M=D^\bot\circ\hat M_{\mathit{ex}}\).
Since \(D^\bot\) is a complete MSC and \(\hat M\) is complete,
Lemma~\ref{lem:strip-complete-prefix} yields that \(\hat M_{\mathit{ex}}\) is
a complete MSC with each \(e_X\in\ltr{X}{\proj{X}{E}}\), hence
\(\hat M_{\mathit{ex}}\in\sem{\mathcal D_E}\).

\emph{Inductive step \(k>0\).}
Since \(k>0\), the owner \(B\)'s component begins with \(d_B^\top\) (a \(\top\)-decision block).
By label compatibility (Definition~\ref{def:msc}) and FIFO ordering on each \(B\to X\) control
channel, every recipient \(X\)'s component of \(\hat M\) also begins with \(d_X^\top\);
idle agents have \(d_X^\top=\varepsilon\).
The tuple \(D^\top=(d_X^\top)_X\) is a complete MSC by the same matching argument.
By Table~\ref{tab:projection-while}, every complete local trace of
\(\proj{X}{P}\) that starts with a \(\top\)-decision has the form
\(d_X^\top\cdot s_X\) with \(s_X\in\ltr{X}{\proj{X}{Q;P}}\): after the
\(\top\)-decision block the body \(Q\) executes completely, followed by the
recursive continuation \(P\).
Write the \(X\)-component of \(\hat M\) as \(d_X^\top\cdot s_X\)
(with \(s_X=\varepsilon\) for idle agents where \(d_X^\top=\varepsilon\)) and
define \(N\df(s_X)_X\), so \(\hat M=D^\top\circ N\).
Since \(D^\top\) is a complete MSC and \(\hat M\) is complete,
Lemma~\ref{lem:strip-complete-prefix} yields that \(N\) is a complete MSC
with \(s_X\in\ltr{X}{\proj{X}{Q;P}}\) for every \(X\).

Since \(s_X\in\ltr{X}{\proj{X}{Q;P}}\), the semantics of sequential
composition gives a factorization \(s_X=b_Xw_X\) with
\(b_X\in\ltr{X}{\proj{X}{Q}}\) and \(w_X\in\ltr{X}{\proj{X}{P}}\).
The factorization is unique by prefix-freeness of \(\ltr{X}{\proj{X}{Q}}\)
(Lemma~\ref{lem:local-prefix-free}).

Apply Lemma~\ref{lem:uniform-zipper} to \(Q\) with \(U=(b_X)_X\) and
\(V=(w_X)_X\).
The prefix condition \(b_X\in\lpref{X}{\proj{X}{Q}}\) holds since
\(b_X\in\ltr{X}{\proj{X}{Q}}\subseteq\lpref{X}{\proj{X}{Q}}\).
The zipper side condition
\(w_X\neq\varepsilon\Longrightarrow b_X\in\ltr{X}{\proj{X}{Q}}\)
holds trivially since \(b_X\in\ltr{X}{\proj{X}{Q}}\) for every \(X\).
The tuple \((b_Xw_X)_X=N\) is an MSC by the above.
The lemma yields \(\bar U\) with \(\mathsf{ZipPost}_Q(U,\bar U,V)\).
Since \(b_X\in\ltr{X}{\proj{X}{Q}}\) for every \(X\),
Corollary~\ref{cor:U-complete} gives that \(\hat M_1\df(b_X)_X\) is a complete MSC,
hence \(\hat M_1\in\sem{\mathcal D_Q}\).

Since \(N=\hat M_1\circ W\) with \(W\df(w_X)_X\) and both \(N\) and \(\hat M_1\) complete
MSCs, Lemma~\ref{lem:strip-complete-prefix} yields that \(W\) is a complete
MSC.
For each $X$, \(w_X\in\ltr{X}{\proj{X}{P}}\), so \(W\in\sem{\mathcal D_P}\).
The owner's component in \(W\) contains exactly \(k-1\) \(\top\)-decision
blocks: the unique decomposition of \(\hat M\)'s owner component consists of
the first block \(d_B^\top\), the complete body trace
\(b_B\in\ltr{B}{\proj{B}{Q}}\), followed by the unique decomposition of
\(w_B\), which therefore contains exactly \(k-1\) outer \(\top\)-decision
blocks.
Applying the induction hypothesis to \(W\) yields
\[
W=\bigl(D^\top\circ\hat M_2\bigr)\circ\cdots\circ
\bigl(D^\top\circ\hat M_k\bigr)\circ\bigl(D^\bot\circ\hat M_{\mathit{ex}}\bigr).
\]
Prepending \(D^\top\circ\hat M_1\) gives
\[
\hat M=D^\top\circ N=D^\top\circ\hat M_1\circ W
=\bigl(D^\top\circ\hat M_1\bigr)\circ\cdots\circ\bigl(D^\top\circ\hat M_k\bigr)\circ\bigl(D^\bot\circ\hat M_{\mathit{ex}}\bigr),
\]
which is the desired factorization.
\qed
\end{proof}

\subsection{Proof of Theorem~\ref{thm:correctness}}
Fix a well-typed global workflow $P$ over lifelines $\mathscr L$.
We prove the two items.

\medskip
\noindent
\textbf{(1) Completeness of realization.}
We show: for every $M\in\sem{P}$ there exists $\hat M\in\sem{\mathcal D_P}$
with $\mathsf{erase}(\hat M)=M$.
Proceed by structural induction on $P$.

\smallskip
\emph{Base cases: $\varepsilon$, $\mathtt{msg}$, $\mathtt{act}$.}
If $P=\varepsilon$, take $\hat M=M_{\varepsilon}$.
If $P$ is $\mathtt{msg}$/$\mathtt{act}$, then
$\proj{X}{P}$ produces exactly the corresponding local letters on the
endpoints/owner and $\varepsilon$ elsewhere, so the unique global one-step MSC
(for that statement) is realized directly.
No control letters occur, and thus $\mathsf{erase}$ is the identity.

\smallskip
\emph{If-construct.}
Let $P=\mathtt{if}\ c@B\ \mathtt{then}\ Q_\top\ \mathtt{else}\ Q_\bot$ and
$M\in\sem{P}$.
Then either $M=M_{\IfT{c}{B}}^{B}\circ M_\top$ with $M_\top\in\sem{Q_\top}$
or $M=M_{\IfF{c}{B}}^{B}\circ M_\bot$ with $M_\bot\in\sem{Q_\bot}$.
Assume the $\top$-case (the other case is analogous).
By IH applied to $Q_\top$, obtain $\hat M_\top\in\sem{\mathcal D_{Q_\top}}$
with $\mathsf{erase}(\hat M_\top)=M_\top$.
Let $\hat M\df D^\top\circ \hat M_\top$, where $D^\top$ is the complete
decision block consisting of the owner-side choice event on $B$, the
projection-introduced control sends on $B$, and the matching control receives
on all recipients \(A\in(\LS{Q_\top}\cup\LS{Q_\bot})\setminus\{B\}\).
The prepended decision block is a complete MSC, and \(\hat M_\top\) is a
complete MSC by IH because it lies in \(\sem{\mathcal D_{Q_\top}}\). Hence
Lemma~\ref{lem:concat-msc} yields that \(\hat M\) is a complete MSC.
Therefore \(\hat M\in\sem{\mathcal D_P}\) by definition of
the projections \(\proj{X}{P}\)
and the local trace semantics, and erasing the control letters yields exactly
$M_{\IfT{c}{B}}^{B}\circ M_\top=M$.

\smallskip
\emph{While-construct.}
Let $P=\mathtt{while}\ c@B\ \mathtt{do}\ Q\ \mathtt{exit}\ E$
and $M\in\sem{P}$.
Then $M$ has the shape
\[
M = \left(M_{\WhT{c}{B}}^{B}\circ M_1\right)\circ\cdots\circ
    \left(M_{\WhT{c}{B}}^{B}\circ M_k\right)\circ
    \left(M_{\WhF{c}{B}}^{B}\circ M_{\mathit{ex}}\right)
\]
for some $k\ge 0$, where each $M_i\in\sem{Q}$ and
$M_{\mathit{ex}}\in\sem{E}$.
Apply IH to each $M_i$ and $M_{\mathit{ex}}$ to obtain realizations
$\hat M_i\in\sem{\mathcal D_{Q}}$ and
$\hat M_{\mathit{ex}}\in\sem{\mathcal D_{E}}$ that erase to the
corresponding global MSCs.
In each body iteration prepend, on $B$, the choice letter $\WhT{c}{B}$ followed
by the projection-introduced decision broadcasts to all $A\in\mathcal{R} = (\LS{Q}\cup \LS{E})\setminus\{B\}$.
For the exit prepend $\WhF{c}{B}$ on $B$ and the corresponding broadcasts.
(When $\mathcal{R}=\emptyset$ there are no broadcasts, but the choice letter is
still prepended on $B$.)
Concatenate the resulting MSCs.
Each prepended decision block is a complete MSC, and each realized body/exit
fragment is a complete MSC by IH. Repeated application of
Lemma~\ref{lem:concat-msc} therefore shows that the concatenation is a
complete MSC. By construction this yields $\hat M\in\sem{\mathcal D_P}$.
Erasure removes the newly inserted $\kappa_{\mathsf{ctrl}}^P$-broadcasts and,
by the induction hypotheses, maps each realized body and exit fragment to its
global counterpart, yielding exactly $M$.

\smallskip
\emph{Composition.}
Let $P=P_1; P_2$ and $M\in\sem{P}$.
Then $M=M_1\circ M_2$ with $M_1\in\sem{P_1}$ and $M_2\in\sem{P_2}$.
By IH, obtain $\hat M_1\in\sem{\mathcal D_{P_1}}$ and
$\hat M_2\in\sem{\mathcal D_{P_2}}$ such that $\mathsf{erase}(\hat M_i)=M_i$.
Let $\hat M \df \hat M_1 \circ \hat M_2$.
Since \(\hat M_1\) and \(\hat M_2\) are complete MSCs, Lemma~\ref{lem:concat-msc}
shows that \(\hat M\) is again a complete MSC. By the definition of $;$ in the
MSC semantics, $\hat M\in\sem{\mathcal D_{P_1; P_2}}=\sem{\mathcal D_P}$.
Moreover $\mathsf{erase}(\hat M)=M$.

\smallskip

This concludes~(1).

\medskip
\noindent
\textbf{(2) Soundness of realization.}
We show: for every $\hat M\in\sem{\mathcal D_P}$,
$\mathsf{erase}(\hat M)\in\sem{P}$.
Proceed by structural induction on $P$.

\smallskip
\emph{Base cases and control cases.}
 For $P=\varepsilon$ and for $P$ equal to $\mathtt{msg}$/$\mathtt{act}$,
 projection introduces no control letters and
 $\mathsf{erase}$ is the identity. The result follows immediately from the
 correspondence between projection and the canonical one-step MSCs.
 
 For the if-case, let
 \[
 P=\mathtt{if}\ c@B\ \mathtt{then}\ Q_\top\ \mathtt{else}\ Q_\bot.
 \]
 By Lemma~\ref{lem:control-decompose}, every
 \(\hat M\in\sem{\mathcal D_P}\) has the form
 \[
 \hat M = D^\nu \circ \hat M_\nu
 \]
 for a unique \(\nu\in\{\top,\bot\}\) and some
 \(\hat M_\nu\in\sem{\mathcal D_{Q_\nu}}\).
 Within \(D^\nu\), erasure removes the control broadcasts, leaving
 exactly the corresponding global choice MSC
 \(M_{\IfT{c}{B}}^{B}\) or \(M_{\IfF{c}{B}}^{B}\).
 Since \(\hat M\in\sem{\mathcal D_P}\) is a complete MSC,
 Lemma~\ref{lem:erase-preserves} gives that \(\mathsf{erase}(\hat M)\) is a
 complete MSC.
 By IH, \(\mathsf{erase}(\hat M_\nu)\in\sem{Q_\nu}\). Hence
 \[
 \mathsf{erase}(\hat M)\in\sem{P}.
 \]

 For the while-case, let
 \[
 P=\mathtt{while}\ c@B\ \mathtt{do}\ Q\ \mathtt{exit}\ E.
 \]
 By Lemma~\ref{lem:control-decompose}, every
 \(\hat M\in\sem{\mathcal D_P}\) decomposes as
 \[
 \hat M =
 \bigl(D^\top\circ \hat M_1\bigr)\circ\cdots\circ
 \bigl(D^\top\circ \hat M_k\bigr)\circ
 \bigl(D^\bot\circ \hat M_{\mathit{ex}}\bigr)
 \]
 for some \(k\ge 0\), with each
 \(\hat M_i\in\sem{\mathcal D_{Q}}\) and
 \(\hat M_{\mathit{ex}}\in\sem{\mathcal D_{E}}\).
 Within the decision blocks \(D^\top,D^\bot\), erasure removes the control
 broadcasts, leaving the
 corresponding owner choice events
 \(M_{\WhT{c}{B}}^{B}\) and \(M_{\WhF{c}{B}}^{B}\).
 By IH,
 \[
 \mathsf{erase}(\hat M_i)\in\sem{Q}
 \qquad\text{and}\qquad
 \mathsf{erase}(\hat M_{\mathit{ex}})\in\sem{E}.
 \]
 Since \(\hat M\in\sem{\mathcal D_P}\) is a complete MSC,
 Lemma~\ref{lem:erase-preserves} gives that \(\mathsf{erase}(\hat M)\) is a
 complete MSC.
 Therefore \(\mathsf{erase}(\hat M)\) has exactly the global while-shape from
 Definition~\ref{def:inductive-msc}, so \(\mathsf{erase}(\hat M)\in\sem{P}\).

 \smallskip
 \emph{Composition.}
Let $P=P_1; P_2$ and $\hat M=(w_X)_X\in\sem{\mathcal D_P}$.
Then we have $w_X\in\ltr{X}{\proj{X}{P_1};\proj{X}{P_2}}$, so factor
$w_X=u_Xv_X$ with $u_X\in\ltr{X}{\proj{X}{P_1}}$ and
$v_X\in\ltr{X}{\proj{X}{P_2}}$.
Thus the zipper side condition
$v_X\neq\varepsilon \Longrightarrow u_X\in\ltr{X}{\proj{X}{P_1}}$
holds trivially.
Since $\hat M$ is complete, it is in particular an MSC.
Thus $\lhd_{\hat M}$ is defined and the causality order is acyclic
(Definition~\ref{def:msc}), so Lemma~\ref{lem:uniform-zipper} applies.
Applying Lemma~\ref{lem:uniform-zipper} to
$\hat M=(u_Xv_X)_X$ yields that $V=(v_X)_X$ is an MSC.
Also, $U=(u_X)_X$ is a complete MSC by Corollary~\ref{cor:U-complete}
(since each $u_X\in\ltr{X}{\proj{X}{P_1}}$), hence $U\in\sem{\mathcal D_{P_1}}$.
Since $U$ and $\hat M=U\circ V$ are both complete MSCs, Lemma~\ref{lem:strip-complete-prefix} yields that $V$ is a complete MSC.
Moreover $V\in\sem{\mathcal D_{P_2}}$ since each $v_X\in\ltr{X}{\proj{X}{P_2}}$ by the sequential composition semantics.
By IH, $\mathsf{erase}(U)\in\sem{P_1}$ and $\mathsf{erase}(V)\in\sem{P_2}$.
Finally, $\mathsf{erase}(\hat M)=\mathsf{erase}(U)\circ\mathsf{erase}(V)$, hence
we have $\mathsf{erase}(\hat M)\in\sem{P_1; P_2}=\sem{P}$.

\smallskip

This concludes (2).
\qed

\subsection{Proof of Corollary~\ref{cor:deadlock-free}}
Let \(M=(u_X)_X\in\sempref{\mathcal D_P}\). By definition of prefix semantics,
\(u_X\in\lpref{X}{\proj{X}{P}}\) for every \(X\in\mathscr L\), and \(M\) is an MSC.

Apply Lemma~\ref{lem:uniform-zipper} with \(U\df(u_X)_X\) and
\(V\df(\varepsilon)_X\).
The zipper side condition is vacuous since \(v_X=\varepsilon\) for all \(X\).
Hence the lemma yields \(\bar U=(\bar u_X)_X\) with
\(\mathsf{ZipPost}_P(U,\bar U,V)\).

By Condition~1, \(u_X\preceq\bar u_X\) and
\(\bar u_X\in\ltr{X}{\proj{X}{P}}\) for every \(X\).
By Condition~2, \(\bar U\) is a complete MSC, hence \(\bar U\in\sem{\mathcal D_P}\), and
\[
M=(u_X)_X \preceq (\bar u_X)_X=\bar U.
\]
This is exactly Definition~\ref{def:deadlock-prefix}.
\qed

\end{document}